\documentclass[11pt]{article}

\usepackage{jheppub}
\usepackage{mmacells}
\mmaDefineMathReplacement[≤]{<=}{\leq}
\mmaDefineMathReplacement[≥]{>=}{\geq}
\mmaDefineMathReplacement[≠]{!=}{\neq}
\mmaDefineMathReplacement[→]{->}{\to}[2]

\usepackage{listings}
\usepackage{amsfonts}
\usepackage{amsmath}
\usepackage{amssymb}
\usepackage{amsthm}
\usepackage[mathscr]{eucal}
\usepackage{bbold}
\usepackage{braket}
\usepackage{color}
\usepackage{dsfont}
\usepackage{framed}
\usepackage{mathtools}
\usepackage{physics}
\usepackage[normalem]{ulem}
\usepackage{tensor}
\usepackage{thmtools}
\usepackage{thm-restate}
\usepackage{ulem}
\usepackage{tikz}
\usetikzlibrary{arrows.meta,arrows} 
\usepackage{centernot}
\usepackage{enumitem} 
\usepackage[a]{esvect}  
\usepackage{slashed}
\usepackage{booktabs}
\usepackage{colortbl}

\usepackage[only,llbracket,rrbracket]{stmaryrd} 

\usepackage{yfonts}
\usepackage{float}
\usetikzlibrary{math} 
\usetikzlibrary{calc}		

\usepackage{cleveref}

\newtheorem{conjecture}{Conjecture}

\newtheorem{lemma}{Lemma}
\newtheorem{thm}{Theorem}

\theoremstyle{definition}

\newtheorem{eg}{Example}

\newcommand{\R}{\mathbb{R}}

\newcommand{\N}{{\sf N}}

\newcommand{\I}{X}
\newcommand{\J}{Y}

\newcommand{\QQ}{\mathbf{Q}}
\newcommand{\XX}{A}
\newcommand{\YY}{B}

\newcommand{\ANR}[2]{#1_{\downarrow #2}}
\newcommand{\ANRC}[2] {\overline{{#1}_{\downarrow #2}}}

\newcommand{\fnr}[1]{{f_{\ANR{}{#1}}}}
\newcommand{\mb}{\mathbf}

\newcommand{\wdef}{h}

\newcommand{\tOmega}{\tilde\Omega}

\newcommand{\ent}[1]{\mathsf{S}(#1)} 
\newcommand{\mi}{{\sf I}}  

\newcommand{\be}{\begin{equation}}
\newcommand{\ee}{\end{equation}}

\tikzstyle{startstop} = [rectangle, rounded corners, minimum width=3cm, minimum height=1cm,text centered, draw=black, fill=red!30]

\tikzstyle{process} = [rectangle, minimum width=3cm, minimum height=1cm, text centered, draw=black, fill=orange!30]
\tikzstyle{decision} = [diamond, minimum width=3cm, minimum height=1cm, text centered, draw=black, fill=green!30]

\definecolor{brg}{RGB}{70, 255,70}
\definecolor{darkgreen}{rgb}{0.1, 0.6, 0.2}

\preprint{BRX-TH-6726}

\title{\boldmath A new characterization of the holographic entropy cone}

\author[a]{Guglielmo Grimaldi,}
\author[a,b]{Matthew Headrick,} 
\author[c]{and Veronika E.\ Hubeny}

\affiliation[a]{Martin Fisher School of Physics, Brandeis University, Waltham MA 02453, USA}
\affiliation[b]{Institut des Hautes Etudes Scientifiques, 91440 Bures-sur-Yvette, France}
\affiliation[c]{Center for Quantum Mathematics and Physics (QMAP)\\ Department of Physics \& Astronomy, University of California, Davis CA 95616, USA}
\emailAdd{ggrimaldi@brandeis.edu}
\emailAdd{headrick@brandeis.edu}
\emailAdd{veronika@physics.ucdavis.edu}

\abstract{Entanglement entropies computed using the holographic Ryu-Takayanagi formula are known to obey an infinite set of linear inequalities, which define the so-called RT entropy cone. The general structure of this cone, or equivalently the set of all valid inequalities, is unknown. It is also unknown whether those same inequalities are also obeyed by entropies computed using the covariant Hubeny-Rangamani-Takayanagi formula, although significant evidence has accumulated that they are. Using Markov states, we develop a test of this conjecture in a heretofore unexplored regime. The test reduces to checking that a given inequality obeys a certain majorization property, which is easy to evaluate. We find that the RT inequalities pass this test and, surprisingly, \emph{only} such inequalities do so. Our results not only provide strong new evidence that the HRT and RT cones coincide, but also offer a completely new characterization of that cone.
\newline
\newline
A video abstract is available at \url{https://youtu.be/stYAA9BwJj8}.
}

\begin{document}

\maketitle

\section{Introduction}\label{sec:introduction}

Entanglement has played a central role in holography since the original holographic entanglement entropy proposal by Ryu-Takayanagi (RT) \cite{Ryu:06b20v, Ryu:2006ef} and its covariant generalization by Hubeny-Rangamani-Takayanagi (HRT) \cite{Hubeny:2007xt}. The study of the entanglement structure of \emph{geometric states}, i.e.\ holographic states whose dual description is in terms of a classical geometry, may help elucidate the emergence of the bulk spacetime from its dual CFT description. Constraints on the entanglement structure are realized in the form of linear inequalities on the CFT subsystem entropies. The set of all such inequalities for the entropies computed by the RT formula is formally captured by the so-called \emph{holographic entropy cone} \cite{Bao:2015bfa} --- better called the \emph{RT cone} --- and an impressive amount of work has been dedicated in understanding its elusive structure \cite{Marolf:2017shp,Hubeny:2018trv,Hubeny:2018ijt,Cui:2018dyq,HernandezCuenca:2019wgh,He:2019ttu,Hernandez-Cuenca:2019jpv,Bao:2020zgx,He:2020xuo,Walter:2020zvt,Bao:2020mqq,Akers:2021lms,Bao:2021gzu,Avis:2021xnz,Czech:2021rxe,Fadel:2021urx,Hernandez-Cuenca:2022pst,Li:2022jji,Czech:2022fzb,He:2022bmi,Hernandez-Cuenca:2023iqh,He:2023rox, Czech:2023xed,He:2023cco,He:2023aif,Czech:2024rco,Bao:2024obe,Bao:2024vmy,Bao:2024azn,Czech:2025tds, Czech:2025jnw,Bao:2025sjn}.

Despite much progress, the story is far from complete. Broadly speaking, there remain two important objectives:
\begin{enumerate}
    \item Understand the full structure of the set of RT inequalities, or equivalently the RT cone.
    \item Determine whether the HRT entropies obey the same inequalities, or equivalently if the HRT cone equals the RT cone.
\end{enumerate}
The first is crucial if we hope to utilize all of the information hidden in the holographic entropy cone to learn new things about the physics of holography. The second is important for two reasons. First, it is needed to confirm whether the lessons we draw from objective (1) apply to generic geometric states, not just static\footnote{By static we also allow for \emph{instantaneously static} states, i.e.\ states with a moment of time-reflection symmetry, with the boundary subregions also restricted to lie on the time-reflection symmetric slice.} ones where the RT formula applies. Second, when moving to the covariant setting, both the existence and the consistency of the HRT formula require dynamical input from the gravitational theory. Properties of holographic entanglement entropy (including, but not limited to, the inequalities) become general relativity theorems. Energy conditions, equations of motion, and asymptotic boundary conditions all play crucial roles. This is a hint that the covariant holographic entropy cone  knows about valuable dynamical properties of the holographic map, so meeting objective (2) is of great importance. Strong evidence, coming from several different directions, has already accumulated that indeed the two cones are equal \cite{Bao:2018wwd, Erdmenger:2017gdk,Caginalp:2019mgu, Czech:2019lps,Grado-White:2024gtx, Bousso:2024ysg} (see \cite{Grado-White:2025jci} for a summary).

In this paper, we make a fresh attack on the question of whether the HRT entropies obey the same inequalities as the RT ones, and along the way discover a completely new way to characterize those inequalities. Our strategy will be to look at non-static configurations that saturate a given RT inequality, and then attempt to perturb the state in order to violate the inequality. When the saturation occurs because the same HRT surfaces appear on the two sides of the inequality, it will be robust under perturbations. Instead, we seek configurations where the HRT surfaces are different, but happen to have the same total area, so that a perturbation of the bulk geometry has a chance of leading to a violation. We find this phenomenon in so-called Markov-state configurations \cite{Casini:2017roe}, where the CFT is in the vacuum in Minkowski space and all the regions lie on a common light cone. This will lead us to the notion of a \emph{null reduction}, which involves deleting every term in an inequality where a particular party does not appear. For example, for the MMI inequality,
\be
\mathsf{S}(AB)+\mathsf{S}(AC)+\mathsf{S}(BC)\ge \mathsf{S}(A)+\mathsf{S}(B)+\mathsf{S}(C)+\mathsf{S}(ABC)\,,
\ee
null reducing on $A$ yields the strong subadditivity inequality
\be\label{SSA0}
\mathsf{S}(AB)+\mathsf{S}(AC)\ge \mathsf{S}(A)+\mathsf{S}(ABC)\,.
\ee
Null reductions are studied in section \ref{sec:sats-and-null}, where it is observed that null reducing any RT inequality yields an RT inequality.

In section \ref{sec:bulk-perturbations}, we use the Einstein equation and null energy condition to show that robustness of a given inequality under a certain class of perturbations of the bulk metric in a light-cone configuration requires its null reductions to obey a certain property. Specifically, with the null-reduced inequality written so that all terms on both sides have coefficient 1, we turn the LHS into a list of numbers, one for each term, where $A$ is replaced by $a$, $AB$ by $a+b$, etc., and similarly for the RHS; we then ask whether, for every set of positive values of $a,b,\ldots$, the LHS list is majorized\footnote{
The definition and basic properties of majorization are briefly reviewed in appendix \ref{sec:review-majorization}. Intuitively one may think of majorization akin to mixing: The vector $\vec x$ is majorized by $\vec y$, written $\vec x\prec\vec y$, if $\vec x$ lies in the convex hull of all permutations of $\vec y$.} by the RHS list. For example, for \eqref{SSA0}, we have $(a+b,a+c)$ for the LHS and $(a,a+b+c)$ for the RHS; indeed, for any positive $a,b,c$, the former pair is majorized by the latter pair (since it is less dispersed). We call this the ``majorization test''. The inequality is valid for any light-cone configuration (LCC), within the given class of states, if and only if it passes the majorization test for all null reductions; we then say it is \emph{LCC-safe}. Indeed, MMI is LCC-safe, as expected since we know it is valid for HRT  \cite{Wall_2014}.

In section \ref{sec:results}, we find that the LCC-safe property has several remarkable features. First, it seems to be obeyed by all RT inequalities. We tested all known inequalities up to 6 regions \cite{Hernandez-Cuenca:2023iqh}, as well as members of the two known infinite families \cite{Czech:2023xed} up to 13 regions. The test thereby provides very strong evidence, in a completely new regime, that the HRT cone equals the RT cone. 
Notice that even though the LCC class of configurations is rather special in the collection of all possible configurations and states, it is quite general in the non-trivial saturation context; in other bulk geometries which are not pure AdS (or quotients thereof like BTZ), we would not encounter this potential violation danger zone.
In the absence of some other kind of saturating configuration that is not trivially robust under perturbations, the only remaining possibility for a counterexample would be a violation deep inside the cone --- some kind of island of violation. We suspect this scenario to be unlikely.

Second, the converse statement appears to hold as well: every LCC-safe inequality is an RT inequality.\footnote{In the Note Added at the end of Introduction we provide a brief update rendering this statement obsolete.  However, in order to maintain historical accuracy, we have chosen to not propagate these ``living updates" through the paper.
} This was a completely unexpected result, from our viewpoint. It says that any inequality that can be violated in a static state can also be violated in a light-cone configuration. We do not know why this should be the case. It implies that LCC-safety amounts to a completely new characterization of the RT cone. We also note that, as a practical matter, LCC-safety is extremely fast to check for any given inequality, much faster than existing methods for finding entropy inequalities.

This work thus addresses both objectives (1) and (2). It also opens several directions for future study. First, it would be valuable to prove that an inequality is LCC-safe if and only if it is an RT inequality. Because it is so fast to check, this property can be used to speed the discovery of new inequalities. We can also ask what the new perspective tells us about the physical content and implications of the holographic inequalities. In other words, what is special about the entanglement structure of geometric states in quantum gravity?
Although our reformulation in terms of majorization may at a first glance seem like a strange mathematical property, its origins harking back to light-cone configurations suggest the hope of a more operational interpretation, and open an intriguing relation to bulk causality.
Finally, and most ambitiously, we can hope that validity of inequalities in the light-cone regime can be promoted to validity in general time-dependent states and configurations, thereby proving that the HRT and RT cones are equal.

First, we begin in section \ref{sec:review} with a brief review of the necessary background material concerning RT inequalities, which we will also use as an opportunity to set our notation and terminology.

\paragraph{Note added:} Since the appearance of this paper, a follow-up work by the same three authors and an additional author has appeared \cite{Grimaldi:2026lbq}, in which Conjectures \ref{conj:sHIQmaj} and \ref{conj:QisTrue} are proven, while Conjectures \ref{conj:majsHIQ} and \ref{conj:allNR} are disproven by explicit counterexamples. That work establishes further combinatorial properties of null reductions. See also \cite{Czech:2026zca} for an independent proof of Conjectures \ref{conj:sHIQmaj} and \ref{conj:QisTrue}.

\section{Review of the RT cone}
\label{sec:review}

Let us begin by discussing entropy cones in quantum mechanics. Fix a state $\rho$ of a quantum system on $\mathsf{N}$ parties\footnote{In quantum mechanics, these can be individual components of an interacting system. In quantum field theory, they are spatial subregions on some Cauchy slice.} $\{A_1,A_2,\dots, A_{\mathsf{N}}\}$, with the Hilbert space factorizing $\mathscr{H} = \mathscr{H}_{1} \otimes \mathscr{H}_2 \otimes \cdots \otimes \mathscr{H}_\mathsf{N}$. The entropies of all possible combinations of the parties form a real vector with $2^{\mathsf{N}}-1$ components, i.e. 
\begin{equation}
    \vec{\mathsf{S}} = \left( \mathsf{S}(A_1),\, \mathsf{S}(A_2),\, \dots,\,\mathsf{S}(A_1A_2),\, \mathsf{S}(A_1A_3),\, \dots,\, \mathsf{S}(A_1A_2\dots A_{\mathsf{N}})\right) \in \mathbb{R}^{2^\mathsf{N}-1}\,.
\end{equation}
This vector lives in \emph{entropy space} and the set of all physically realizable entropy vectors for all possible density matrices forms a convex cone in entropy space known as the \emph{quantum entropy cone}. While it is clear from its construction that the cone is symmetric under permutation of the $\mathsf{N}$ parties, the cone actually enjoys a larger \emph{purification symmetry} under permutations of $\mathsf{N} +1$ parties, where the $(\mathsf{N}+1)$st party, the \emph{purifier}, is defined as the complement subsystem $(A_{1}A_2\dots A_\mathsf{N})^c$ which purifies the state.\footnote{Recall that for a pure state $\mathsf{S}(A_{1} A_2 \dots A_\mathsf{N+1}) = 0$. Then for any subsystem $\I$ its entropy equals the one of its complement $\I^c$, 
$\mathsf{S}(\I) = \mathsf{S}(\I^c)$. If the purifier is treated democratically as any other party, then the larger permutation symmetry is manifest.}

An \emph{entropy cone} can be viewed as the intersection of half-spaces defined by linear inequalities on the entropies. Each facet of the cone is contained in (and spans) a hyperplane associated with saturation of one of these inequalities. The full collection of inequalities characterizes the system under consideration; while the \emph{quantum} entropy cone characterizes the set of all physical states of any quantum system, the more restricted \emph{holographic} entropy cone pertains to geometric states in holography, or more specifically CFT states admitting a dual description in terms of a classical geometry for which the entropies of spatial regions are computed by the RT formula.\footnote{
    One may also apply this framework to abstract constructs not restricted to physical systems; for example the subadditivity cone specified simply by non-negativity of mutual information has played a key role in several recent explorations of the holographic entropy cone 
    \cite{Hernandez-Cuenca:2022pst,He:2022bmi,He:2023cco,He:2023aif,He:2024xzq,Hubeny:2024fjn}.
}  To emphasize that most of these inequalities have been proved only in the static context, in the present work we refer to the holographic entropy cone as the \emph{RT cone}.

The RT cone is known fully up to $\mathsf{N}=5$ \cite{HernandezCuenca:2019wgh}, whereas the quantum entropy cone is known only up to $\mathsf{N} = 3$.  
At $\mathsf{N} = 3$ the quantum entropy cone is delimited by subaditivity (SA), or equivalently non-negativity of mutual information (MI),
\begin{equation}\label{eq:SA}
    \mi(A:B) := \mathsf{S}(A) + \mathsf{S}(B) - \mathsf{S}(AB) \geq 0,
\end{equation}
and strong subaditivity (SSA), or equivalently non-negativity of conditional mutual information (CMI),
\begin{equation}\label{eq:SSA}
    \mi(A:C|B) :=  
    \mi(AB:C) - \mi(B:C) = 
    \mathsf{S}(AB) + \mathsf{S}(BC) - \mathsf{S}(B)  - \mathsf{S}(ABC) \geq 0,
\end{equation}
and their permutations in the $\mathsf{N} +1$ parties. 
The holographic entropy cone is likewise specified by SA \eqref{eq:SA} at $\mathsf{N} =2$, but becomes more restricted for larger $\mathsf{N} $.  Already at $\mathsf{N} =3$, we have a new inequality, called the  monogamy of mutual information (MMI) \cite{Hayden:2011ag}, which can be expressed as the non-positivity of tripartite information,
\begin{equation}\label{eq:MMI}
    -\mi_3(A:B:C)    :=  \ent{AB} + \ent{BC}+ \ent{AC} - \ent{A}- \ent{B}- \ent{C}  - \ent{ABC}\geq 0\,.
\end{equation}
Note that MMI is not a true inequality in quantum mechanics, since it can be violated by some quantum states such as the 4-party GHZ state. In holography, it renders SSA redundant, and the RT cone for $\mathsf{N} = 3,4$ is in fact fully determined by SA and MMI. 

We will call an inequality that holds for RT entropies a \emph{static holographic entropy inequality} (sHEI).\footnote{In much of the literature on the holographic entropy cone, such inequalities are simply called holographic entropy inequalities (HEIs). Here we include the prefix \emph{s} in order to emphasize that these inequalities have so far only been proved for RT.} A special role is played by the so-called \emph{primitive} sHEIs, which cannot be written as conical combinations of other sHEIs. Primitive sHEIs determine facets of the holographic entropy cone. For $\mathsf{N}\le5$, the full set of primitive sHEIs, and therefore the entire holographic entropy cone, is known \cite{Bao:2015bfa, HernandezCuenca:2019wgh}. For $\mathsf{N} = 6$, a large set of primitive sHEIs is known, but it is likely quite incomplete \cite{Hernandez-Cuenca:2023iqh}. 
In addition, two infinite families of primitive sHEIs for arbitrarily large $\mathsf{N}$ have been constructed \cite{Czech:2024rco}.

Currently, the only method available to determine whether a given inequality is an sHEI is by showing the existence of a so-called contraction map \cite{Bao:2015bfa}. Consider an inequality of the form LHS $\ge$ RHS, where LHS and RHS are conical combinations of the entropies:
\be\label{LHSRHS}
\text{LHS}=\sum_{l=1}^Lc_l \, \mathsf{S}(\I_l)\,,\qquad
\text{RHS}=\sum_{r=1}^Rd_r \, \mathsf{S}(\J_r)\,,
\ee
where $\I_l$, $\J_r$ are composite parties, and $c_l, d_r$ are strictly positive. We encode these terms into bit strings, known as occurrence vectors; for each region $A_i$, we define the bit string $\vec x_i\in \{0,1\}^{L}$ by $(x_i)_l =1$ if $A_i \subseteq \I_l$ and $(x_i)_l =0$ otherwise, and similarly for $\vec y_i\in \{0,1\}^{R}$. The purifier region is therefore assigned the bit strings $\vec x_{\mathsf{N}+1} =\vec{0}$ and $\vec y_{\mathsf{N}+1} = \vec{0}$. A \emph{contraction map} is a map $f: \{0,1\}^{L} \to \{0,1\}^R$ such that, for all $i = 1, \ldots,\mathsf{N}+1$,
\begin{equation}
\label{eq:contrOV}
f(\vec x_i) = \vec y_i\,,
\end{equation}
and for all pairs of bit strings $\vec x,\vec x' \in \{0,1\}^L$,
\begin{equation}
\label{eq:contrcond}
\norm{f(\vec x) - f(\vec x')}_{\text{rhs}} \leq \norm{\vec x-\vec x'}_{\text{lhs}}\,,
\end{equation}
where $\norm{\cdot}_{\rm lhs}$ is the Hamming norm with weights defined by the coefficients $c_l$, and similarly for $\norm{\cdot}_{\rm rhs}$ with weights $d_r$. It can be shown \cite{Bao:2015bfa} that, if a contraction map exists, then the inequality is an sHEI. The contraction map proof method can be understood as a formal combinatorial rewriting of the original geometric inclusion-exclusion argument \cite{Headrick:2007km}, where the map $f$ is a bookkeeping device that defines the correct inclusions and exclusions to perform to prove the inequality. It has also been argued that the converse holds: every sHEI is provable by a contraction map \cite{Bao:2025sjn}.

\subsection{Information quantities}

As we saw above with the MI, CMI, and negative tripartite information, it is convenient to introduce the notion of an \emph{information quantity}, namely a linear combination of entropies:
\begin{equation}
    \mathbf{Q}(\vec{\mathsf{S}}) = \sum_\I Q_\I \, \mathsf{S}(\I)\,,
\end{equation}
where $\I$ runs over all $2^{\mathsf{N}}-1$ non-empty subsystems. Any linear entropy inequality can be written in the form
\be\label{eq:sHIQsHEI}
\mathbf{Q}(\vec{\mathsf{S}}) \geq 0
\ee
for some information quantity $\QQ$. As in \eqref{LHSRHS}, the inequality \eqref{eq:sHIQsHEI} can also be written
\be
\text{LHS}\ge\text{RHS}\,,
\ee
where by definition LHS and RHS have positive coefficients. Whenever we refer to LHS and RHS of a given information quantity, we will always mean them in this sense.

A \emph{static holographic information quantity} (sHIQ) is one for which \eqref{eq:sHIQsHEI} is a sHEI, in other words, for which \eqref{eq:sHIQsHEI} holds for any static geometric state and any specification of $\mathsf{N}$ spatial regions. Any conical combination of sHIQs is an sHIQ, so the sHIQs form a cone in the dual space to the entropy space; specifically, the dual cone to the RT cone. A \emph{primitive} sHIQ, corresponding to a primitive sHEI, is one that cannot be written as a conical combination of other sHIQs; such an sHIQ is an extreme ray of the dual cone. All known primitive sHIQs can be written with integer coefficients 
\cite{Bao:2015bfa,HernandezCuenca:2019wgh,Hernandez-Cuenca:2023iqh,Czech:2024rco}.

Two properties enjoyed by most primitive sHIQs, that will play an important role in our work, are balance and superbalance. An information quantity $\QQ$ is \emph{balanced} if, for every party $A_i$,
\begin{equation}
    \sum_{\I \supseteq A_i} Q_\I = 0\,.
\end{equation}
Balance implies that, when evaluated on a configuration of non-adjoining regions, $\QQ(\vec{\mathsf{S})}$ is finite and scheme-independent, as the UV divergences associated with the boundaries cancel. The MI, CMI, and tripartite information are all balanced. An example of a non-balanced sHIQ is the Araki-Lieb quantity,
\begin{equation}\label{ALIQ}
    \mathbf{Q}_{\text{AL}}    := \ent{A} + \ent{AB} - \ent{B}\,.
\end{equation}
In fact, all primitive sHIQs are balanced except instances of Araki-Lieb.

One way to get an sHIQ from another one (on the same set of parties) is to permute the parties $A_1,\ldots,A_{\N}$. This also obviously preserves the balance condition. More generally, one can permute those parties \emph{along with the purifier} $O$. After doing so, one should use purity to rewrite the quantity using just $A_1,\ldots,A_{\N}$. For example, for $\N=2$, starting from the MI $\ent{A}+\ent{B}-\ent{AB}$, one can exchange $B$ with $O$ to obtain $\ent{A}+\ent{O}-\ent{AO}$, then rewrite the last two terms in terms of $A,B$; the result is precisely the Araki-Lieb quantity \eqref{ALIQ}. As we see from that example, such a permutation does not necessarily preserve the balance condition. We call the set of all information quantities obtained from a given one by permutations, including those involving the purifier, its \emph{orbit}.

We say that $\QQ$ is \emph{superbalanced} if all members of its orbit are balanced. This is equivalent to being balanced and having vanishing total coefficient for any \emph{pair} of parties $A_iA_j$:
\be
\sum_{X\supseteq A_iA_j}Q_X=0\,.
\ee
It is also equivalent to the quantity vanishing for any state whose purification consists entirely of Bell pairs. The tripartite information is superbalanced, but the MI and CMI are not. It turns out that all primitive sHIQs are superbalanced except instances of MI 
\cite{He:2020xuo}.
 
Properties like superbalance are very helpful in restricting the possible forms that sHIQs can take. A particularly interesting one is the so-called \emph{tripartite form} \cite{Hernandez-Cuenca:2023iqh}. This form automatically incorporates several conditions such as superbalance \cite{He:2020xuo}, and there is evidence that any primitive sHIQ (except instances of MI) can be cast in this form \cite{Hernandez-Cuenca:2023iqh,HubenyLiuWIP}. The tripartite form comprises conditional tripartite informations with a definite sign: 
\begin{equation}\label{eq:tripartite-form}
    \mathbf{Q} = \sum_i - \mi_3(X_i:Y_i:Z_i|W_i),
\end{equation}
where $X_i, Y_i, Z_i$, and $W_i$ are any disjoint subsystems (individual or composite) and where the conditional tripartite information is defined analogously to the CMI:
\begin{equation}\label{eq:cond-tripartiteI}
    \mi_3(X:Y:Z|W) := \mi_3(WX:Y:Z) - \mi_3(W:Y:Z)\,.
\end{equation}
Here $W_i$ can be the empty set, in which case $\mi_3(X_i:Y_i:Z_i|\emptyset)$ reduces to the tripartite information $\mi_3(X_i:Y_i:Z_i)$. The rigidity of the tripartite form \eqref{eq:tripartite-form} makes it a powerful tool for generating new inequalities, as demonstrated in \cite{Hernandez-Cuenca:2023iqh}, and we will use it in \cref{sec:TFproperties} to analyze properties of null-reduced sHIQs.

\section{Saturation of inequalities and null reductions}\label{sec:sats-and-null}

We seek to test whether the HRT formula obeys the same set of inequalities as the RT formula. Our strategy is to construct non-static configurations that saturate a given sHEI, and then perturb the bulk geometry slightly to try to violate the inequality. (In entropy space, saturating an inequality means the entropy vector is sitting on a facet of the holographic entropy cone. A violation would correspond to the vector being pushed out of the cone.) Writing the sHEI in the form $\text{LHS}\ge\text{RHS}$, there are two natural ways a configuration of HRT surfaces can saturate the inequality:
\begin{enumerate} 
\item The HRT surfaces that calculate the entropies on the LHS are the same as the surfaces that calculate those on the RHS. This kind of saturation is robust under perturbations of the bulk. Under a small perturbation of the metric, the same surfaces will continue to calculate the two sides, so the inequality will continue to be saturated.\footnote{An exception may occur if the configuration happens to be sitting on a phase transition for one of the sides. It would be interesting to explore whether this can be a source of tests for the covariant validity of RT inequalities.} Such a configuration therefore does not pose a threat to the validity of the inequality.

\item Different HRT surfaces appear in LHS and RHS, but they happen to have the same total area. Then, we can try to change the metric in order to increase the area of the surfaces in RHS without changing those in LHS, leading to a violation of the inequality. Naively it seems we will always be able to make a counterexample in this way. However, the perturbation may violate the null energy condition (NEC) and therefore be illegal in the classical limit where we are working. Indeed, the examples we will study have the feature that they saturate NEC, so some small perturbations of the metric do violate it. The question then is whether there exist NEC-preserving perturbations that lead to a counterexample to the given inequality. This seems to be an excellent place to test for violations, and we will focus on this second class of configurations. 
\end{enumerate}

\subsection{Light-cone configurations}
\label{sec:light-cone}

We wish to construct configurations of the second type. For two different sets of HRT surfaces to have the same total area requires a very special geometry. The only situations we know of where this occurs systematically are those in which the field theory is in Minkowski space (or a conformally flat spacetime), the state is the vacuum, and the regions lie on a common light cone.

\begin{figure}
    \centering
    \includegraphics[width=0.8\linewidth]{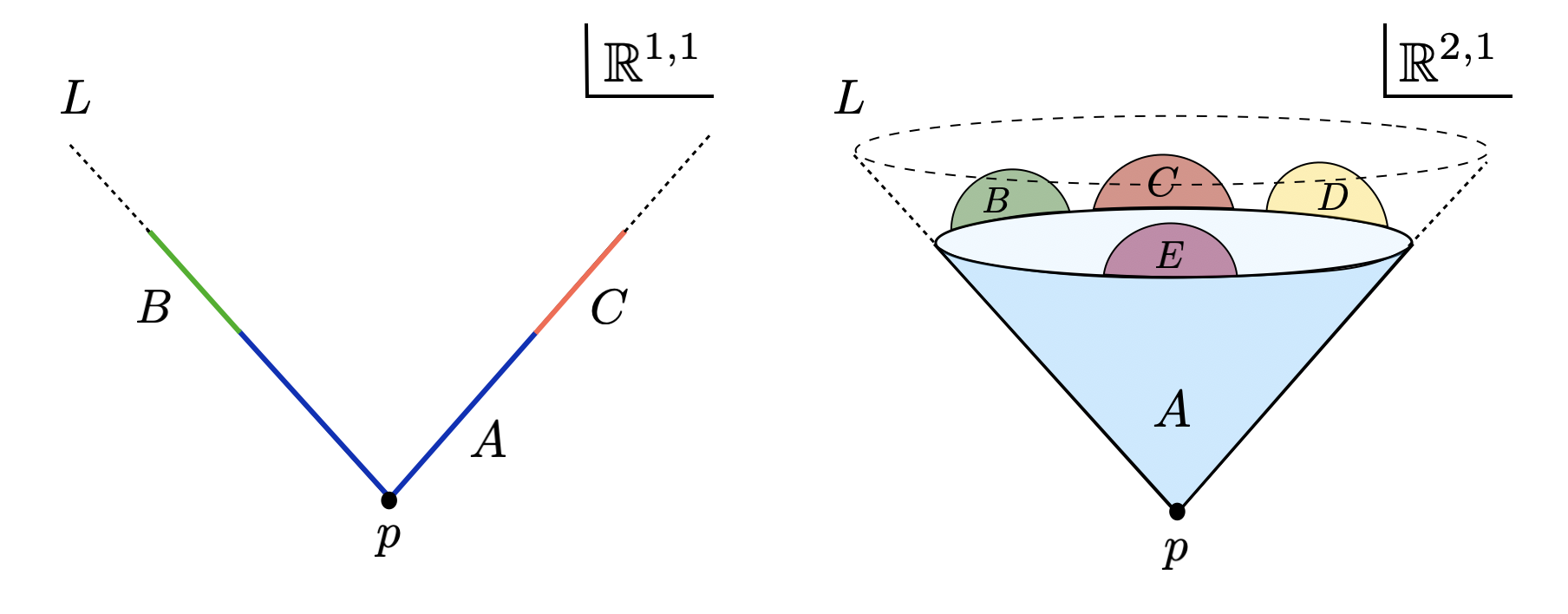}
    \caption{\textbf{Left}: LCC with central region $A$ for three boundary regions $A,B$ and $C$ in 1+1 dimensional Minkowski space. In any CFT, this configuration saturates SSA, $\mi(B:C|A)=0$, implying that the state reduced on these three regions is a Markov state. \textbf{Right}: LCC in 2+1 dimensional Minkowski space for five boundary regions.}
    
    \label{fig:light-cone-config}
\end{figure}

Let $\mathbf{Q}$ be a superbalanced information quantity on $\mathsf{N}$ parties. We are \emph{not} assuming at this point that $\QQ$ is an sHIQ. We arrange all of the regions appearing in $\mathbf{Q}$ to lie on the future light-cone $L$ of a point $p$ in $d$-dimensional Minkowski space, $L:=\partial I_{\rm bdy}^+(p)$, with one of the regions, say $\XX$, containing $p$, and the others adjoining $\XX$. The arrangement must be such that each null ray of $L$ intersects $\XX$ on a single finite interval, and intersects at most one other region, on a single finite interval touching $\XX$; thus, following the ray from $p$, it starts in $\XX$, leaves $\XX$, possibly immediately enters another region, leaves that region, and does not enter any regions again. Thus, the null rays of $L$ can be partitioned according to which region other than $\XX$, if any, they intersect. We call such a configuration of regions a \emph{light-cone configuration (LCC)}, and the region containing $p$ the \emph{central region}. The simplest example of an LCC can be constructed with three regions in $d=2$; see figure \ref{fig:light-cone-config}, left side. With more than three regions, we must have $d\ge3$ to construct an LCC; see figure \ref{fig:light-cone-config}, right side.

\begin{figure}
    \centering
    \includegraphics[width=0.35\linewidth]{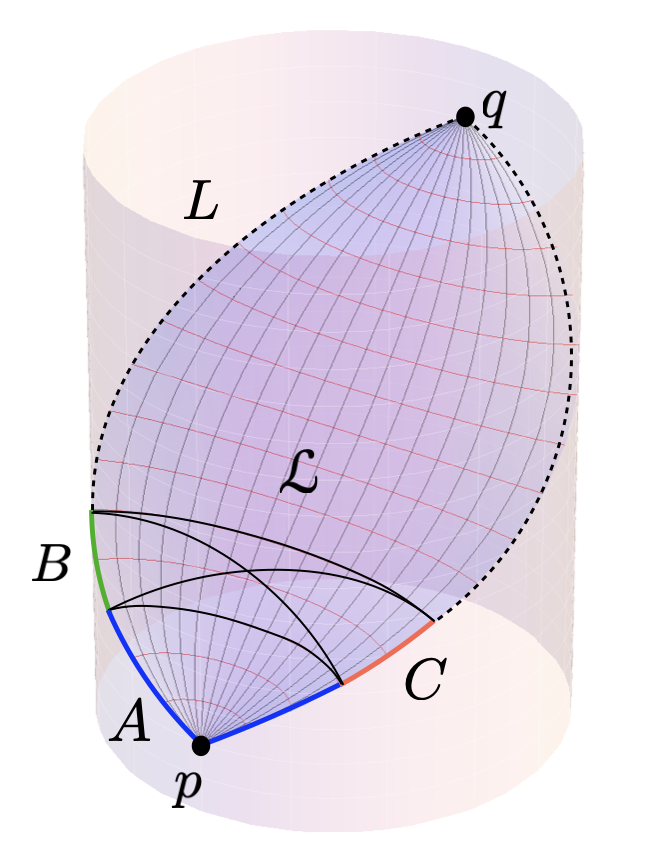}
    \caption{LCC with three regions $A$, $B$ and $C$ as in the left figure \ref{fig:light-cone-config}, on $\R\times S^1$, with the light-cone $L$ (dashed) extended until it closes at the point $q$. The interior of the cylinder is global AdS$_3$, and the future light cone in the bulk is $\mathcal{L}$, with its null generators likewise converging at $q$. As an example we consider MMI: $\ent{AB} + \ent{AB}+ \ent{AC}  \geq \ent{A} + \ent{B} + \ent{C} + \ent{ABC}$. Since $A$ is the region containing the tip of the light-cone, studying MMI in this configuration is equivalent to studying its null reduction on $A$, which is SSA: $\ent{AB} + \ent{AC} \geq \ent{A} + \ent{ABC}$. We draw in black the four RT surfaces contributing for each term;     the configuration saturates the inequality since every bulk null generator (thin blue lines) intersects every RT surface exactly once.
    }
    \label{fig:light-cone-surfaces}
\end{figure}

It will be convenient for us to embed our Minkowski space conformally in $\R\times S^{d-1}$. Then $L$ can be continued until it closes at the point $q$ that is antipodal to $p$ on the $S^{d-1}$ and at a time $\pi R$ later, with $R$ the radius of the $S^{d-1}$; see figure \ref{fig:light-cone-surfaces} (whose bulk part is described at the end of the next subsection). The full light cone $L$ forms a Cauchy slice for the boundary spacetime, with the purifying region $O$ occupying the complement of the regions $A,B,\ldots$ and containing $q$.

\subsection{Saturation of inequalities}
\label{sec:saturation}

The entropy vector $\vec{\mathsf{S}}$ depends on both the configuration of regions and the state. Since in this paper we will consider only LCCs, we will not indicate the configuration explicitly, but will write $\vec{\mathsf{S}}_\psi$ to indicate the dependence on the state $\psi$, with $\vec{\mathsf{S}}_{\rm vac}$ representing the entropy vector in the vacuum. In this subsection, we will show that $\QQ(\vec{\mathsf{S}}_{\rm vac})=0$ --- again, assuming only that $\QQ$ is superbalanced.

When evaluating $\QQ$ on an LCC, there is a fundamental difference between the terms that include the central region (say $\XX$) and those that don't. Namely, for a composite region containing $\XX$, and therefore containing $p$, the state has the form of a Markov state, which means that the modular Hamiltonian takes a particularly simple form \cite{Casini:2017roe}; on the other hand, for a region not containing $\XX$, the state is degenerate. We will be more specific about both cases below. Because of this distinction, we partition $\mathbf{Q}$ accordingly,
\be\label{Qdecomp}
\QQ=\ANR{\mathbf{Q}}{\XX}+\ANRC{\QQ}{\XX}\,,
\ee
where $\ANR{\QQ}{\XX}$ includes all the terms that contain $\XX$ and $\ANRC{\QQ}{\XX}$ includes the rest. The fact that $\QQ$ is balanced means that the sum of the coefficients in $\ANR{\QQ}{\XX}$ vanishes. We also have:

\begin{lemma}\label{thm:balance}
Both $\ANR{\QQ}{\XX}$ and $\ANRC{\QQ}{\XX}$ are balanced.
\end{lemma}
\begin{proof}
Consider a region other than $\XX$, say $\YY$. Since $\QQ$ is superbalanced, the sum of the coefficients of the terms in $\QQ$ containing both $\XX$ and $\YY$ vanishes. These are the same as the terms in $\ANR{\QQ}{\XX}$ containing $\YY$. So $\ANR{\QQ}{\XX}$ is balanced. Since the sum of the coefficients of all the terms in $\QQ$ containing $\YY$ vanishes, so does the sum of the coefficients of the terms containing $\YY$ but not $\XX$. These are the terms in $\ANRC{\QQ}{\XX}$ containing $\YY$. So $\ANRC{\QQ}{\XX}$ is also balanced.
\end{proof}

We will now argue that $\ANR{\QQ}{\XX}$ and $\ANRC{\QQ}{\XX}$ separately vanish on $\vec{\mathsf{S}}_{\rm vac}$. Actually, for $\ANRC{\QQ}{\XX}$, we will show something stronger: it vanishes on $\vec{\mathsf{S}}_\psi$ for any finite-energy state $\psi$. Thus, its vanishing is simply a consequence of the LCC, while that of $\ANR{\QQ}{\XX}$ depends on the state being the vacuum. We begin with $\ANRC{\QQ}{\XX}$.

For each term in $\ANRC{\QQ}{\XX}$, the corresponding region is null and therefore does not have a well-defined entropy, as its proper size in the null direction is less than the UV cutoff; in holography, this is reflected in the fact that the HRT surface is degenerate and hugs the boundary.\footnote{One might be tempted to assign such null regions an entropy of 0. However, this is not correct, as can be seen by considering a null triangle with sides $A,B$. $AB$ is a spacelike region and therefore has a positive entropy. If we assigned entropy 0 to $A$ and $B$, we would get a contradiction with subadditivity. Instead, we should simply say that it has an undefined entropy. We thank T. Takayanagi for pointing this out.} We can regulate their entropies by making $L$ slightly spacelike. The key fact we will use is that, for a composite region, the entropies add up: the MIs all vanish, as the size of each region is very small compared to their separation; in holography, the (degenerate) HRT surface is the union of those of the elementary regions. Since $\ANRC{\QQ}{\XX}$ is balanced, it vanishes on $\vec{\mathsf{S}}_\psi$. Note that this statement does not depend on the state $\psi$ (for finite-energy states); indeed, in holography the degenerate HRT surfaces hugging the boundary are insensitive to the metric in the interior of the bulk.

Having shown that
\be\label{QAcomp}
\ANRC{\QQ}{\XX}(\vec{\mathsf{S}}_\psi)=0
\ee
for any state, we will henceforth focus on the more interesting quantity $\ANR{\QQ}{\XX}$, which we call the \emph{null reduction} of $\QQ$ on $\XX$. The null reduction on any given region is a linear projection operator on the dual entropy space.

We now argue that
\be\label{QASvac}
\ANR{\QQ}{\XX}(\vec{\mathsf{S}}_{\rm vac})=0\,.
\ee
For each term in $\ANR{\QQ}{\XX}$, we can appeal to the Markov property, which says that a region $\I$ on a null cone containing the vertex $p$ has the following modular Hamiltonian \cite{Casini:2017roe}:
\be\label{Markov}
H_{\rm vac}(\I)=2\pi\int_{S^{d-2}} d\Omega\int_0^{\lambda_\I(\Omega)}d\lambda\,\lambda^{d-1} \, \frac{\lambda_\I(\Omega)-\lambda}{\lambda_\I(\Omega)} \, T_{\lambda\lambda}\,,
\ee
where $\Omega$ is the angle on $L$, $\lambda$ the affine parameter along the null ray with $p$ at $\lambda=0$, $\lambda_\I(\Omega)$ the value of $\lambda$ at the boundary of $\I$, and $T_{\lambda\lambda}$ the stress tensor component in the $\lambda$ direction. This implies that the entropy can be written in the following form:
\be\label{S1}
\ent{\I} = \int_{S^{d-2}} d\Omega\,s(\lambda_\I(\Omega))\,,
\ee
where $s$ a function that can be read off from \eqref{Markov} but whose form is not important for us. Denoting the non-central regions by $\YY_i$ ($i=1,\ldots,\N-1$), we can write \eqref{S1} as
\be\label{SX}
\ent{\I}= \int_{R_0}d\Omega\,s(\lambda_\XX(\Omega))+
\sum_{B_i\not\subseteq X}\int_{R_i}d\Omega\,s(\lambda_\XX(\Omega))+\sum_{B_i\subseteq X}\int_{R_i}d\Omega\,s(\lambda_i(\Omega))\,,
\ee
where $R_i\subset S^{d-2}$ is the set of angles subtended by $\YY_i$, $R_0:=S^{d-2}\setminus(\cup_{i=1}^{\N-1}R_i)$ is the set of angles not subtended by any of the $\YY_i$, and $\lambda_i(\Omega):=\lambda_{\XX\YY_i}(\Omega)$ is the outer boundary of $\YY_i$. For the null-reduced information quantity $\ANR{\QQ}{\XX}$, we then have
\begin{align}
\ANR{\QQ}{\XX}(\vec{\mathsf{S}}_{\rm vac}) &=\sum_{\I}q_\I\ent{\I} \nonumber\\
&=\int_{R_0}d\Omega\left(\sum_\I q_\I\right)s(\lambda_\XX(\Omega)) \nonumber\\
&\qquad\qquad+\sum_i\int_{R_i}d\Omega\left(\sum_{\I\not\supseteq\YY_i}q_\I\right)s(\lambda_\XX(\Omega))
+\sum_i\int_{R_i}d\Omega\left(\sum_{\I\supseteq\YY_i}q_\I\right)s(\lambda_i(\Omega)) \nonumber\\
&=0\,,
\end{align}
where $q_\I$ is the coefficient of $\I$ in $\ANR{\QQ}{\XX}$ (i.e.\ $q_\I=Q_\I$ if $X\supseteq\XX$ and 0 otherwise), the second equality follows from \eqref{SX} and rearranging terms, and the last equality follows from the fact that $\ANR{\QQ}{\XX}$ is balanced, so each sum in large parentheses vanishes.

The above argument did not refer to holography and holds in any CFT. However, it is instructive to see why $\ANR{\QQ}{\XX}(\vec{\mathsf{S}}_{\rm vac})=0$ in holography. Since we are in the vacuum, the bulk is AdS, and all the HRT surfaces lie on $\mathscr{L}=\partial I_{\rm bulk}^+(p)$, the future light cone in the bulk of $p$. In fact, each HRT surface intersects every null generator of $\mathscr{L}$ exactly once. Furthermore, $\mathscr{L}$ has vanishing expansion, so the area element is constant along any geodesic. So any geodesic makes a contribution to $\ANR{\QQ}{\XX}$ proportional to the sum of all the coefficients in $\ANR{\QQ}{\XX}$, which (again by balance) vanishes. See figure \ref{fig:light-cone-surfaces} for a visualization of the HRT surfaces that saturate MMI as an example.\footnote{To be slightly more careful, we should take care of the UV regulator. We are assuming a regulator that, for every null ray of $\mathcal{L}$, either includes or excludes the entire ray. For example, we could use a bulk cutoff surface that is a timelike hyperplane in Poincar\'e coordinates that intersects the boundary on a spacelike boundary hyperplane containing $p$. The important point is that, while the entropy vector $\vec{\mathsf{S}}_{\rm vac}$ depends on the regulator, since $\QQ$ is superbalanced, $\QQ(\vec{\mathsf{S}}_{\rm vac})$ does not.}

In the next section, we will perturb the bulk metric to try to violate the inequality. Before doing so, in the rest of this section, we will present some examples and general properties of null reductions of superbalanced information quantities. We will see that they are quite interesting objects in their own right.

\subsection{Examples of null reductions} 
 
To gain some intuition about null reductions, in this subsection we will give a few examples of null reductions of sHIQs.
\begin{eg}
Consider MMI:
\begin{equation}
    \mathbf{Q}_{\text{MMI}} = \ent{AB} +\ent{BC} + \ent{AC} - \ent{A} - \ent{B} - \ent{C} - \ent{ABC} \,.
\end{equation}
Its three null reductions are
\begin{align}
\ANR{\QQ^{^\text{MMI}}}{A}  = \ent{AB} + \ent{AC} - \ent{A} - \ent{ABC} &= \mi(B:C|A)\label{MMINRA}\\
    \ANR{\QQ^{^\text{MMI}}}{B} = \ent{AB} + \ent{BC} - \ent{B} - \ent{ABC} &= \mi(A:C|B)\\
    \ANR{\QQ^{^\text{MMI}}}{C} = \ent{BC} + \ent{AC} - \ent{C} - \ent{ABC} &= \mi(A:B|C)\,.
\end{align}
Notice that each null reduction is a non-negative information quantity, by SSA.
\end{eg}
\begin{eg}\label{Q5eg}
Consider the following primitive sHIQ on five regions
\begin{multline}
    \mathbf{Q}^{[5]} = \ent{ABC} + \ent{ABD}  + \ent{ACE}  + \ent{BCD} + \ent{BCE}\\
    - \ent{A} - \ent{BC} - \ent{BD} - \ent{CE} - \ent{ABCD} - \ent{ABCE}\,.
\end{multline}
Here and throughout the paper we reference sHIQs from the list of known sHIQs publicly available in \cite{hecdata}; the superscript in the above corresponds to the position of the information quantity in the list. Its null reductions are:
\begin{align}
    \ANR{\QQ^{[5]}}{A} &= \ent{ABC} + \ent{ABD}  + \ent{ACE} - \ent{A}   - \ent{ABCD} - \ent{ABCE}\nonumber \\
    &= \mi(C:D|AB) + \mi(B:CE|A) \label{Q5NRA} \\
    \ANR{\QQ^{[5]}}{B} &= \ent{ABC} + \ent{ABD} + \ent{BCD} + \ent{BCE}\nonumber\\  &\qquad\qquad\qquad\qquad\qquad- \ent{BC} - \ent{BD}-\ent{ABCD}-\ent{ABCE} \nonumber \\
&    = \mi(A:E|BC) + \mi(A:C|BD)\\
    \ANR{\QQ^{[5]}}{C} &= \ent{ABC} + \ent{ACE}  + \ent{BCD} +
    \ent{BCE}\nonumber\\
    &\qquad\qquad\qquad\qquad\qquad- \ent{BC} - \ent{CE} - \ent{ABCD}- \ent{ABCE}\nonumber \\
    &= \mi(A:D|BC) + \mi(A:B|CE)\\
    \ANR{\QQ^{[5]}}{D} &= \ent{ABD}  + \ent{BCD} - \ent{BD} - \ent{ABCD} \nonumber \\
    &= \mi(A:C|BD)\\
    \ANR{\QQ^{[5]}}{E} &= \ent{ACE} + \ent{BCE} - \ent{CE}  - \ent{ABCE} \nonumber\\
    &= \mi(A:B|CE)\,.
\end{align}
Again, each expression is non-negative by SSA.
\end{eg}

Based on these two examples, the reader may be tempted to guess that null reduction of an sHIQ always yields a positive sum of CMIs. This is not true, as it can be seen in the following example:
\begin{eg}\label{Q7eg}
Here is another 5-party primitive sHIQ:
\begin{align}\label{eq:Q7}
   \mathbf{Q}^{[7]} =  & \,\ent{AD}+\ent{BC} + \ent{ABE}+\ent{ACE}+\ent{ADE}+\ent{BDE}+\ent{CDE}\nonumber\\
   &\qquad-\ent{ABDE}-\ent{ACDE}-\ent{BCE} \nonumber\\
&\qquad   -\ent{AE}-\ent{DE}
   -\ent{A}-\ent{B}-\ent{C}-\ent{D}\,.
\end{align}
Null-reducing on $E$, we find
\begin{align}
\ANR{\QQ^{[7]}}{E}&=\ent{ABE}+\ent{ACE}+\ent{ADE}+\ent{BDE}+\ent{CDE}\nonumber \\
&\qquad\qquad -\ent{ABDE}-\ent{ACDE}-\ent{BCE}-\ent{AE}-\ent{DE} \nonumber
\\ \label{eq:Q72}
&=\mathbf{Q}^{[7]} +\mi(B:C)+\mi(A:D)\,.
\end{align}
While the null reduction in this case is not a sum of CMIs, it is still an sHIQ. Indeed, we find this to be a general feature. Based on extensive tests, we conjecture that the null reductions of any superbalanced sHIQ are sHIQs. We also find that the converse statement holds: a superbalanced information quantity is an sHIQ \emph{only if} all of its null reductions are. These conjectures and the evidence for them will be discussed in sec.\ \ref{sec:results}.
\end{eg}

\section{The majorization test}\label{sec:bulk-perturbations}

Let us recap what we have done so far. We have considered configurations that saturate holographic entropy inequalities. We have learnt that configurations that are most susceptible to potentially being violated under a bulk perturbation are the LCCs from subsection \ref{sec:light-cone}. We have seen that when evaluated on such configurations, the sHEIs reduce to their null reduction on some given party. The question that remains to be answered is: can we violate such an inequality by perturbing the bulk slightly? To answer the question we first need to understand how these surfaces look on the bulk light-cone $\mathscr{L}$, and then how their areas change under bulk metric perturbations. Restricting bulk perturbations to those that respect the null energy condition will lead us to the majorization test.

\subsection{HRT surfaces}
\label{sec:HRTsurfaces}

The metric for the Poincar\'e patch of AdS${}_{d+1}$ ($d\ge2$) can be written as
\be
ds^2 = \frac1{\cos^2\theta}\left(-2du\,dv-u^2\,dv^2+d\theta^2+\sin^2\theta\,d\Omega_{d-2}^2\right),
\ee
where $u\in(-\infty,0)$, $v\in\R$, $\theta\in[0,\pi/2)$, $\Omega\in S^{d-2}$. These are related to the usual Poincar\'e coordinates by
\be
z=-\frac{\cos\theta}u\,,\qquad
x^0=-\frac1u+v\,,\qquad
\vec x=-\frac{\sin\theta}u\hat x(\Omega)\,,
\ee
where $\hat x(\Omega)$ is the unit vector in $\R^{d-1}$ at the angle $\Omega$. With $p$ at $z=x^0=x^i=0$, its future light cone $\mathscr{L}$ is at $v=0$, with $\theta,\Omega$ specifying the null ray and $u$ an affine parameter along the ray. The boundary is at $\theta=\pi/2$, and the boundary light cone $L$ with the same vertex $p$ is at $v=0$, with $\Omega$ specifying the null ray and $u$ a (non-affine) parameter along the ray. ($\lambda=-1/u$ is an affine parameter for the boundary ray.) Extending both the bulk and boundary light cones until they reconverge at the point $q$ on the boundary, as shown in figure \ref{fig:light-cone-surfaces}, extends $u$ to cover the whole real line.

Given a boundary region $\I\subset L$ containing $p$, with boundary at $u=u^*(\Omega)$, the HRT surface $\gamma_\I$ lies on $\mathscr{L}$, intersecting each ray exactly once, and can therefore be expressed as a function $u(\theta,\Omega)$ obeying the boundary condition $u(\pi/2,\Omega)=u^*(\Omega)$. This function is determined by the following elliptic PDE on the hemisphere parametrized by $\theta,\Omega$:\footnote{Note that the PDE \eqref{HRTPDE} is regular at $\theta=0$, where it locally reduces to the flat Laplacian in spherical coordinates in $d-1$ dimensions. Also, for $d\neq3$, \eqref{HRTPDE} can be written $\tilde\nabla^2u=0$, where $\tilde\nabla$ is the scalar Laplacian with respect to the following metric on the hemisphere:
\be
d\tilde s^2 = (\cos\theta)^{-2(d-1)/(d-3)}(d\theta^2+\sin^2\!\theta \,d\Omega_{d-2}^2)\,.
\ee}
\begin{equation}\label{HRTPDE}
\partial_\theta\left(\frac{\sin^{d-2}\theta}{\cos^{d-1}\theta}\partial_\theta u\right)+\frac{\sin^{d-4}\theta}{\cos^{d-1}\theta}\nabla^2_{\Omega} u=0\,,
\end{equation}
where $\nabla^2_{\Omega}$ is the Laplacian on $S^{d-2}$. The important feature of the PDE \eqref{HRTPDE} for our purposes is that it is linear in $u$. Because of this linearity, the solution $u(\theta,\Omega)$ is a linear functional of the boundary function $u^*(\Omega)$. Furthermore, this functional is strictly monotone: If $u^*_1\ge u^*_2$ for all $\Omega$ with $u_1^*>u_2^*$ for some $\Omega$, then $u_1> u_2$ for all $\theta\in[0,\pi/2)$ and all $\Omega$. This follows from the maximum principle for elliptic PDEs, which implies that if the boundary value $u^*_1-u_2^*$ is everywhere non-negative and somewhere positive then the corresponding solution $u_1-u_2$ is everywhere positive in the interior of the domain.

Given an LCC with central region $\XX$ and other regions $\YY_i$ ($i=1,\ldots,\N-1$), we define the functions
\be
b^*_i(\Omega):=\begin{cases}
u^*_i(\Omega)-u_\XX^*(\Omega)\,,\quad&\Omega\in R_i \\
0\,,\quad&\Omega\not\in R_i
\end{cases}\,,
\ee
where $R_i$ is the subset of $S^{d-2}$ subtended by $\YY_i$, $u_\XX^*(\Omega)$ is the boundary of $\XX$, and $u_i^*(\Omega)$ is the outer boundary of $\YY_i$. (In terms of the boundary affine parameter $\lambda=-1/u$ and the functions $\lambda_\XX$, $\lambda_i$ used in subsection \ref{sec:saturation} to specify the boundaries of the regions, we have $u_\XX^*(\Omega)=-1/\lambda_\XX(\Omega)$, $u_i^*(\Omega)=-1/\lambda_i(\Omega)$.) Then, for any composite region $\I$ containing $p$,
\be
u^*_\I(\Omega)=u_\XX^*(\Omega)+\sum_{\YY_i\subseteq \I}b_i^*(\Omega)\,.
\ee
We now define the functions $u_\XX(\theta,\Omega)$, $b_i(\theta,\Omega)$ as the solutions to \eqref{HRTPDE} with boundary conditions (at $\theta = \pi/2$) $u_\XX^*(\Omega)$, $b_i^*(\Omega)$ respectively. The HRT surface for $\I$ is then given by
\be\label{udef}
u_\I(\theta,\Omega)=u_\XX(\theta,\Omega)+\sum_{\YY_i\subseteq \I}b_i(\theta,\Omega)\,.
\ee
Note that, since $b_i^*$ is non-negative, $b_i$ is positive.

\subsection{Bulk perturbations}

Let's now perturb the bulk and ask how the entropy vector for our LCC changes. Specifically, since the HRT surface is extended in the $\theta$ and $\Omega$ directions, we perturb those parts of the metric:
\begin{equation}\label{eq:perturbation}
\delta g_{\mu\nu} = 2\epsilon \wdef_{\mu\nu}\,,
\end{equation}
where $\wdef_{\mu\nu}$ is a function of all of the coordinates and is non-zero only for $\mu,\nu=\theta,\Omega$, and we take $\epsilon>0$ to characterize the size of the deformation (subsequently treating it as a bookkeeping parameter to focus on the linear order in the perturbation expansion).  Since $\gamma_\I$ is extremal, the change in its area $|\gamma_\I|$ at first order in $\epsilon$ is simply due to the change in the metric at its coordinate position, which is:
\begin{equation}\label{areaelement1}
\delta|\gamma_\I| = \epsilon\int d\theta \,d\Omega\,\frac{\sin^{d-2}\theta}{\cos^{d-1}\theta}\,\wdef\,,
\end{equation}
where $h:=g^{\mu\nu}h_{\mu\nu}$ is evaluated at $u=u(\theta,\Omega)$, $v=0$. Henceforth, since only the value of $h$ on the light-cone is relevant, we fix $v=0$. We also make the abbreviations
\be
\tOmega:=(\theta,\Omega)\,,\qquad
d\tOmega:=d\theta\,d\Omega\,\frac{\sin^{d-2}\theta}{\cos^{d-1}\theta}\,,
\ee
so \eqref{areaelement1} becomes
\be\label{areaelement}
\delta|\gamma_\I|=\epsilon\int d\tOmega\,\wdef\,.
\ee
Define the function $\vec s(\tOmega)$, valued as an entropy-space vector, by its components
\be
s(\tOmega)_\I:=\begin{cases}
\wdef(u_\I(\tOmega),\tOmega)/4G_{\rm N}\,,\quad&\XX\subseteq \I \\
0\,,\quad&\XX\not\subseteq \I
\end{cases}\,.
\ee
We then have, from \eqref{areaelement},
\be\label{deltaS}
\delta\vec{\mathsf{S}} = \epsilon\int d\tOmega
\,\vec s(\tOmega)\,.
\ee

Now let $\QQ$ be a superbalanced information quantity. We then have:
\be\label{deltaQ}
\QQ(\vec{\mathsf{S}}_\psi)=\ANR{\QQ}{\XX}(\vec{\mathsf{S}}_\psi) = \ANR{\QQ}{\XX}(\delta\vec{\mathsf{S}})+O(\epsilon^2) = \epsilon\int d\tOmega\,\ANR{\QQ}{\XX}(\vec s(\tOmega))+O(\epsilon^2)\,,
\ee
where in the first equality we used \eqref{Qdecomp} and \eqref{QAcomp}, in the second one \eqref{QASvac}, and in the last one \eqref{deltaS} and the linearity of $\ANR{\QQ}{\XX}$.

A couple of notational housekeeping items will make our task less laborious going forward. First, from \eqref{deltaQ}, we see that the different values of $\tOmega$, i.e.\ the different null rays of $\mathscr{L}$, make independent contributions. Therefore, we henceforth fix a value of $\tOmega$ and drop the explicit dependence on it. We should therefore think of $\wdef$ as a function solely of $u$, defined on the interval $(-\infty,0)$. Second, we will asssume that $\QQ$ has integer coefficients. This is not an essential assumption, and we will relax it at the end of the subsection, but it will simplify the presentation. Accordingly, we write the information quantity $\ANR{\QQ}{\XX}$ so that every term has coefficient $\pm1$:
\be\label{InJndef}
\ANR{\QQ}{\XX}=\sum_{n=1}^N\mathsf{S}(\I_n)-\sum_{n=1}^N\mathsf{S}(\J_n)\,,
\ee
where $\I_n$ are the composite regions appearing in $\ANR{\QQ}{\XX}$ with positive coefficients and $\J_n$ are those appearing with negative coefficients. The same region $\I$ will appear multiple times in the list of $\I_n$s if $\mathsf{S}(\I)$ has coefficient greater than 1 in $\ANR{\QQ}{\XX}$, and similarly for the negative coefficients. From the fact that the total of the coefficients of $\ANR{\QQ}{\XX}$ vanishes, we know that the number $N$ of $\I_n$s (not to be confused with the number of elementary regions $\N$) equals the number of $\J_n$s. In this notation, \eqref{deltaQ} tells us that the inequality $\QQ(\vec{\mathsf{S}})\ge0$ is safe against the perturbation \eqref{eq:perturbation} as long as
\be\label{fundamental}
\sum_{n=1}^N \wdef(u_{\I_n})\ge\sum_{n=1}^N \wdef(u_{\J_n})\, ,
\ee
where,  for any region $X$, by \eqref{udef}, $u_\I=u_\XX+\sum_{\YY_i\subseteq \I}b_i$.

Let us see what \eqref{fundamental} looks like in a couple of examples.
\begin{eg}
The null reduction on $A$ of MMI is
\be
\ANR{\QQ^{\rm MMI}}{\XX}= \ent{AB}+\ent{AC}- \ent{A}-\ent{ABC}
\ee
(see \eqref{MMINRA}), so \eqref{fundamental} becomes
\be\label{MMIfundamental}
\wdef(a+b)+\wdef(a+c)\ge \wdef(a)+\wdef(a+b+c)\,,
\ee
where for notational convenience we wrote $u_A$ as $a$, $b_B$ as $b$, etc.
\end{eg}
\begin{eg}
We look at the null reduction on $A$ (as in the right side of figure  \ref{fig:light-cone-config}) of the 5-party sHIQ $\QQ^{[5]}$ discussed in example \ref{Q5eg} above (see \eqref{Q5NRA}):
\be
\ANR{\QQ^{[5]}}{\XX}= \ent{ABC}+\ent{ABD}+\ent{ACE}-\ent{A}-\ent{ABCD}-\ent{ABCE}\,.
\ee
The inequality \eqref{fundamental} then takes the following form:
\begin{align}\label{Q5fundamental}
     \wdef(a + b + c) + \wdef(a+b+d)& + \wdef(a+c+e)\nonumber\\
     &\ge \wdef(a) + \wdef(a+b+c+d) + \wdef(a + b + c + e)\,.
\end{align}
\end{eg}

\subsection{Null energy condition \& majorization}

Let us now take a closer look at \eqref{fundamental}. 
The arguments $u_{\I_n}$ appearing on the left-hand side are all distinct from the ones $u_{\J_n}$ appearing on the right-hand side, else they would have cancelled in $\ANR{\QQ}{\XX}$ (more precisely, they are distinct as functions of  $a,b,c,\ldots$, which means that generically they have distinct values).
Therefore we can easily choose a function $\wdef$ that violates \eqref{fundamental}; for example, we can choose $\wdef$ to be negative at some $u_{\I_n}$ and zero at all other $u_{\I_n}$s and at all $u_{\J_n}$. It would therefore appear to be very easy to construct a non-static counterexample to any sHEI.

We know the argument in the previous paragraph must be incorrect, because we know that HRT entropies obey MMI \cite{Wall_2014,Bousso:2024ysg}. The resolution to this puzzle is that the bulk metric is not arbitrary: since we are working in classical Einstein gravity, it must obey the Einstein equation with matter obeying the null energy condition. Without these conditions, it is indeed easy for the HRT formula to violate MMI (and indeed subadditivity, SSA, and many other consistency conditions \cite{Callan:2012ip}). 
Specifically, the Einstein equation and NEC constrain the function $\wdef$ defining our perturbation \eqref{eq:perturbation} via the focusing lemma, which says that the expasion $\Theta$ of a null geodesic congruence obeys
\be
\Theta' \le -\frac{\Theta^2}{d-1}\,,
\ee
where the derivative is with respect to the affine parameter. (Note that, under the perturbation \eqref{eq:perturbation}, the curves $v=0$, constant $\tilde\Omega$ remain null geodesics, with $u$ as an affine parameter.) The expansion is the logarithmic derivative of the area element; in our case, from \eqref{areaelement}, the area element is $1+\epsilon \wdef$, so
\be
\Theta=\epsilon \wdef'+O(\epsilon^2)\,.
\ee
Thus, focusing implies that $h$ is concave.

Now return to the MMI example \eqref{MMIfundamental}. Recalling that $b$ and $c$ must be positive, the two arguments of $\wdef$ appearing on the LHS, $a+b$, $a+c$, lie in between those appearing on the RHS $a$, $a+b+c$, with the same mean value $a+(b+c)/2$. This implies that the LHS values are \emph{majorized} by the RHS ones,
\be\label{MMImajorize}
(a+b,a+c)\prec (a,a+b+c)\,.
\ee
(See appendix \ref{sec:review-majorization} for the definition of majorization and its key properties. Notice the change in direction of the inequality sign: $\ge$ became $\prec$.) Together with the concavity of $\wdef$, \eqref{MMImajorize} implies that \eqref{MMIfundamental} is true: the MMI inequality is safe (as we already knew).

For the $\QQ^{[5]}$ example \eqref{Q5fundamental}, again we see that the arguments on the LHS and RHS have the same mean value, namely $a+(2b+2c+d+e)/3$. The equality of the mean values is a consequence of balance of the null-reduced information quantity, so it will always hold. Again, for any positive $b,c,d,e$, the arguments on the LHS are majorized by those on the RHS:
\be
(a+b+c,a+b+d,a+c+e)\prec(a,a+b+c+d,a+b+c+e)\,.
\ee
And, again, together with the concavity of $\wdef$, this implies that \eqref{Q5fundamental} is true. So $\QQ^{[5]}$ is safe for any LCC with central region $A$.

The general theorem relating majorization to inequalities obeyed by concave functions is the following:
\begin{thm}[Karamata]
\label{thm:karamata}
Given vectors $(x_1,\ldots,x_N)$, $(y_1,\ldots,y_N)$ such that $\sum x_n=\sum y_n$, $\vec x\prec\vec y$ if and only if, for any concave function $h$, 
\be\label{finequality}
\sum_{n=1}^N \wdef(x_n)\ge\sum_{n=1}^N \wdef(y_n)\,.
\ee
\end{thm}
\noindent For a given information quantity, we thus need to test whether the quantities $u_{\I_n}$ from the LHS are majorized by the ones $u_{\J_n}$ from the RHS. Actually, since the term $u_\XX$ appears in all of the $\I_n$s and $\J_n$s, it has no effect on the majorization condition.\footnote{In the version of the majorization test given in the introduction, for ease of presentation a positive quantity corresponding to the central region was included. Again, since by definition the central quantity appears in every term in the null-reduced inequality, whether such a quantity is included or not has no effect on the majorization property.} Thus, writing $\mb{b}=(b_i)_{i=1}^{\N-1}$ and defining\footnote{Defining the $N\times(\N-1)$ matrices $\mathbf{X}$, $\mb Y$ by
\be
X_{ni} = \begin{cases}0\,,&B_i\notin X_n \\1\,,&B_i\in X_n \end{cases}\,,\qquad
Y_{ni} = \begin{cases}0\,,&B_i\notin Y_n \\1\,,&B_i\in Y_n \end{cases}\,,
\ee
we have $\vec x=\mb X\mb b$, $\vec y = \mb Y\mb b$. In the majorization literature, given matrices $\mb X$, $\mb Y$, requiring $\mb X\mb b\prec\mb Y\mb b$ for all positive vectors $\mb b$ is called \emph{positive directional majorization}, \emph{positive-combinations majorization}, or \emph{price majorization}. See section 15.A of \cite{MR2759813} for discussion, examples, and further references. We thank G. Dahl and P. Shteyner for helpful correspondence on this.}
\be\label{xnyndef}
x_n(\mb{b}):=\sum_{\YY_i\in\I_n}b_i\,,\qquad
y_n(\mb b):=\sum_{\YY_i\in\J_n}b_i\qquad
(n=1,\ldots,N)\,,
\ee
is it sufficient to check whether
\be\label{majorization2}
\vec x(\mb b)\prec\vec y(\mb b)\,.
\ee
(Note that if the term $\mathsf{S}(A)$ appears in $\QQ$, and therefore in $\ANR{\QQ}{\XX}$, then the corresponding component $x_n$ or $y_n$ will be 0.) If \eqref{majorization2} holds for all positive numbers $b_i$, then $\QQ(\mathsf{S}_\psi)\ge0$, at first order in $\epsilon$, for any LCC with central region $\XX$, and we say that $\QQ$ \emph{passes the majorization test for $\XX$}.

In fact, we can reverse the logic to show that the converse holds: By the Karamata theorem, if there exists a set of positive numbers $b_i'$ such that \eqref{majorization2} is false, then, for any number $u_\XX'$, there exists a concave function $\wdef$ such that \eqref{fundamental} is false. Furthermore, for any given value $\tOmega'$ of $\tOmega$, there exists an LCC with central region $\XX$ such that $u_\XX(\tOmega')=u_\XX'$ and $b_i(\tOmega')=b_i'$, which can be found by starting with any LCC and multiplying the functions $u^*(\Omega)$, $b_i^*(\Omega)$ by appropriate numbers. For this configuration, $\QQ(\mathsf{S}_\psi)<0$. Therefore the inequality is safe against this type of counterexample if and only if it passes the majorization test for $\XX$.

We say that $\QQ$ is \emph{LCC-safe} if it passes the majorization test for every elementary region. To recap, the procedure for determining LCC-safety is the following. For each elementary region:
\begin{enumerate}
\item null reduce $\QQ$ on that region;
\item define the vectors $\vec x(\mb{b})$, $\vec y(\mb{b})$ by \eqref{xnyndef}, where the $B_i$ are the other elementary regions and the $X_n$, $Y_n$ are the terms on the LHS and RHS of the null-reduced inequality respectively;
\item test whether \eqref{majorization2} holds for all positive numbers $b_i$.
\end{enumerate}
If this is true for all elementary regions, then $\QQ$ is LCC-safe.

By definition, the LCC-safe property is invariant under $S_\N$ permutations of the parties, but not under $S_{\N+1}$ permutations involving the purifer. If we treat the purifier democratically, then a null reduction should be thought of as depending on \emph{two} regions, which contain the past and future vertices $p,q$ of the light cone $L$ respectively. From this democratic viewpoint, each term in $\QQ$ depends on a bipartition of the $\N+1$ regions, and the terms that are retained in the null reduction are those in which the two regions are separated by the bipartition. The vectors $\vec x(\mb{b})$, $\vec y(\mb{b})$ appearing in the majorization condition \eqref{majorization2} depend on the subsets $X_n$, $Y_n$ and therefore on which region is labelled the ``purifier'' and which one the ``central region''. However, we will now show that the majorization condition itself is actually independent of this choice. If we switch the purifier and central region, then each composite region $X_n$, $Y_n$ appearing in the null-reduced quantity gets replaced by the complementary region:
\be
X_n\to X_n'=\mathcal{A}\setminus X_n\,,\qquad
Y_n\to Y_n'=\mathcal{A}\setminus Y_n\,,
\ee
where $\mathcal{A}$ is the full set of $\N+1$ regions. Therefore, the vectors $\vec x(\mb{b})$, $\vec y(\mb{b})$ appearing in the majorization relation \eqref{majorization2} are replaced by $\vec x'(\mb{b})$, $\vec y{\,}'(\mb{b})$, where
\be
x_n'(\mb{b}) = \sum_{B_i\in X_n'}b_i=\sum_{i=1}^{\N-1}b_i-x_n(\mb b)\,,\qquad
y_n'(\mb{b}) = \sum_{B_i\in Y_n'}b_i=\sum_{i=1}^{\N-1}b_i-y_n(\mb b)\,.
\ee
Since majorization is invariant under negating both vectors, and also invariant under adding a constant to all components of both vectors, \eqref{majorization2} is indeed invariant under this transformation. All in all, for a given orbit, there are thus $\N+1\choose2$ independent majorization tests.

The generalization of the majorization test to an information quantity with non-integer coefficients involves the notion of weighted majorization, reviewed in appendix \ref{sec:review-majorization}. Here we relax the condition that the vectors $\vec x$ and $\vec y$ have the same length; call their lengths $M$, $N$ respectively. We now have weight vectors $\vec\alpha$, $\vec\beta$, where $\vec\alpha$ has length $M$ as $\vec x$ and $\vec\beta$ length $N$. The vectors $\vec\alpha$, $\vec\beta$ must have positive components and obey
\be\label{equaltots}
\sum_{m=1}^M\alpha_m=\sum_{n=1}^N\beta_n\,.
\ee
Furthermore, the weighted sums of $\vec x$ and $\vec y$ must be equal:
\be
\sum_{m=1}^M\alpha_mx_m=\sum_{n=1}^N\beta_ny_n\,.
\ee
The weighted majorization condition $(\vec x,\vec\alpha)\prec (\vec y,\vec\beta)$ is equivalent to the statement that, for any concave function $h$,
\be
\sum_{m=1}^M\alpha_mh(x_m)\ge\sum_{n=1}^N\beta_nh(y_n)\,.
\ee
Given a superbalanced information quantity $\QQ$ and central region $\XX$, we denote by $q_\I$ the coefficient of $\I$ in the null-reduced quantity $\ANR{\QQ}{\XX}$. Let $\vec\alpha$ be the vector whose components are the positive $q_\I$s, $\alpha_m=q_{\I_m}$, and $\vec\beta$ the vector whose components are minus the negative ones, $\beta_n=-q_{\I_n}$; balance of $\ANR{\QQ}{\XX}$ ensures that \eqref{equaltots} is satisfied. Slightly generalizing \eqref{xnyndef}, the vectors $\vec x$, $\vec y$ are now defined by
\be\label{xnyndef2}
x_m(\mb b):=\sum_{\YY_i\in\I_m}b_i\,,\qquad
y_n(\mb b):=\sum_{\YY_i\in\J_n}b_i\,.
\ee
We then have $\QQ(\vec{\mathsf{S}}_\psi)\ge0$ for all concave $h$ if and only if $(\vec x(\mb b),\vec\alpha)\prec (\vec y(\mb b),\vec\beta)$. We say that $\QQ$ passes the majorization test for $\XX$ if this holds for all positive values of the $b_i$, and that $\QQ$ is LCC-safe if it passes the majorization test for every elementary region. Tracing through the definitions, one sees that the condition of being LCC-safe is preserved by conical combinations and limits. Therefore the set of LCC-safe information quantities defines a closed convex cone in the dual entropy space, which we call the \emph{LCC cone}.

\section{Conjectures and evidence}\label{sec:results}

In this section we will present four conjectures concerning null reductions and the majorization test, and give strong empirical evidence in favor of each of them. For all four of them, we assume $\QQ$ is a superbalanced information quantity.

\begin{conjecture}\label{conj:sHIQmaj}
If  $\mathbf{Q}$ is an sHIQ, then it is LCC-safe.
\end{conjecture}

\begin{conjecture}\label{conj:majsHIQ}
If $\mathbf{Q}$ is LCC-safe, then it is an sHIQ.
\end{conjecture}

Conjectures \ref{conj:sHIQmaj} and \ref{conj:majsHIQ} are the central claims of this paper. They tie majorization theory with objectives (1) and (2) listed in the introduction: conjecture \ref{conj:sHIQmaj} guarantees \emph{robustness} of sHEIs in the HRT regime, while conjecture  \ref{conj:majsHIQ} \emph{characterizes} sHEIs through majorization. Together, they amount to the claim that the LCC cone equals the RT cone (or, more precisely, its superbalanced sub-cone).

\begin{conjecture}\label{conj:QisTrue}
If $\mathbf{Q}$ is an sHIQ then all of its null reductions are sHIQs.
\end{conjecture}
\begin{conjecture}\label{conj:allNR}
If all of the null reductions of $\QQ$ are sHIQs, then it is an sHIQ.
\end{conjecture}

These two conjectures represent some kind of self-consistency condition on the dual entropy cone.

\subsection{Computational evidence} 

\paragraph{Conjecture \ref{conj:sHIQmaj}:} To test this conjecture, it is sufficient to test it on the primitive sHIQs. We tested a large number of known primitive sHIQs, finding that they all passed the majorization test. Two strategies were employed: analytic and numerical. The analytic strategy consisted of showing that the inequality \eqref{majmax} holds for arbitrary $t$ and arbitrary positive $b_i$. As an example, for MMI null-reduced on the region $A$, we wish to show that
\be
(b,c)\prec(b+c,0)\,.
\ee
If $b\le c$, we have three cases: $t\in (-\infty,b]$, $t\in[b,c]$, $t\in[c,\infty)$. In each case, \eqref{majmax} is easily confirmed.\footnote{Note that, in this simple case, majorization is also immediately evident from the definition in \eqref{majdef}.} This method quickly becomes tedious to do by hand, but can be implemented in \emph{Mathematica} using \texttt{FullSimplify}; see appendix \ref{sec:code} for details. The analytic method was applied to the 1876 known (non-MI) $\N=6$ primitive sHIQs (but not their images under permutations involving the purifier, which were tested using the numerical method below due to computational speed).

For large inequalities, this analytic method becomes too slow. We can instead carry out the majorization test numerically, by directly implementing the definition \eqref{majdef} for random values of the $b_i$; again, see appendix \ref{sec:code} for details. This method can find counterexamples but cannot prove that the test is passed. However, finding no counterexample after many random draws constitutes convincing evidence. Furthermore, the test is extremely fast. For example, testing all 1876 known $\N=6$ primitive sHIQs \cite{Hernandez-Cuenca:2023iqh} with 100 random draws for each null reduction took only about a minute on a standard laptop, versus several days for the analytic method. The numerical method was applied to the toric and projective plane inequalities \cite{Czech:2024rco} up to $\N=13$ and for the 11253 images involving the purifier of the 1876 known $\mathsf{N} = 6$ primitive sHIQs.

The analytic and numerical results taken together constitute overwhelming evidence in favor of conjecture \ref{conj:sHIQmaj}.

\paragraph{Conjecture \ref{conj:majsHIQ}:} This conjecture is supported by testing its contrapositive: a superbalanced information quantity that is not an sHIQ is not LCC-safe. To construct non-sHIQs, we subtracted superbalanced primitive sHIQs from each other. Specifically, we wrote code to do the following:
\begin{enumerate}
    \item Pick a random pair of primitive sHIQs $\mathbf{Q}_i$ and $\mathbf{Q}_j$.
    \item Permute the $\N$ parties plus the purifier in $\QQ_i$, $\QQ_j$ independently with randomly drawn permutations $\pi,\sigma \in S_{\N+1}$ to obtain new quantities $\pi(\mathbf{Q}_i)$ and $\sigma(\mathbf{Q}_j)$. This is done to democratically test the conjecture on all corners of dual entropy space, avoiding potential selection biases. 
    \item Pick a random integer $n \in [4,10]$ and construct the quantity
    \begin{equation}
        \mathbf{Q}_{\text{false}}: = n \pi(\mathbf{Q}_i) - \sigma(\mathbf{Q}_j).
    \end{equation}
\end{enumerate}
Then, we perform the majorization test on $\mathbf{Q}_{\text{false}}$ using the random numerical method described above. Using the known $\N=6$ primitive sHIQs, we performed more than one million random draws, in every case seeing a failure of the majorization test for at least one null reduction. (In practice, the most common occurrence was for the information quantity to fail the majorization test for all null reductions, but cases where some were passed were also not uncommon.)

We also performed a different exploration of the conjecture by exploiting the simple structure of information quantities written in ``tripartite form'' (see \ref{eq:tripartite-form}), which are guaranteed to be superbalanced. By restricting to IQs whose tripartite form contains \emph{only} conditional tripartite informations, one is guaranteed to obtain a \emph{false} superbalanced inequality. To systematically generate information quantities of this type one can do the following procedure. For a given conditional tripartite information, generate all arguments $\{X, Y, Z | W\}$ from a collection of $\mathsf{N}$ regions. This can be done by considering all partitions of $\{1,2,3,\dots, \mathsf{N}+1\}$ into five non-empty subsets, and then picking the four subsets not involving the purifier. Since the fourth argument $W$ is not symmetric with respect to the others, there are four additional expression that can be created. For $\mathsf{N} = 6$, there are 560 possible conditional tripartite informations. We exhaustively checked \emph{all} IQs whose tripartite form contains one or two conditional tripartite informations, i.e.\ information quantities of the form
\be
-\mi_3(X:Y:Z|W)\,, \qquad -\mi_3(X_1:Y_1:Z_1|W_1) -\mi_3(X_2:Y_2:Z_2|W_2) 
\ee
always seeing the failure of the majorization test for at least one null reduction.

\paragraph{Conjecture \ref{conj:QisTrue}:}
This conjecture was tested by showing that all null reductions of all 372 superbalanced $\N=5$ sHIQs are sHIQs by directly checking that each null reduction sits inside the dual cone. See appendix \ref{sec:code} for some details. We will further demonstrate \cref{conj:QisTrue} for all known sHIQs in \cref{sec:TFproperties} using the fact that they can be written in the tripartite form.

\paragraph{Conjecture \ref{conj:allNR}:} Similarly to conjecture \ref{conj:majsHIQ}, this conjecture is supported by testing its contrapositive: constructing a false superbalanced information quantity $\QQ_{\text{false}}$ and checking that at least one of its null reductions is \emph{not} an sHIQ. Since the HEC is known completely only up to $\mathsf{N} = 5$, we tested the conjecture in this regime. First, we construct $\QQ_{\text{false}}$ in the same way as in conjecture 2. Then, we check whether all of its null reductions sit inside the $\mathsf{N} = 5$ dual cone, by checking whether they can be written as conical combinations of the 372 primitive sHIQs. We tested this for hundreds of randomly generated examples of $\QQ_{\text{false}}$, finding no counterexamples.

\subsection{Tripartite form \& null reduction}
\label{sec:TFproperties}

We can analyze the structure of null-reduced information quantities by utilizing the so-called tripartite form of the sHIQ briefly mentioned in \eqref{eq:tripartite-form}.  Let us first recall the definition.  We say an information quantity $\QQ$ can be written in tripartite form (TF), or equivalently is \emph{TF-compatible}, if it can be expressed as a positive sum of negative tripartite and conditional tripartite informations,
\begin{equation}\label{eq:tripartite-form-def}
    \QQ = \sum_i - \mi_3(X_i:Y_i:Z_i|W_i) \, ,
\end{equation}
where $X_i, Y_i, Z_i$, and $W_i$ are any disjoint subsystems (individual or composite). The special case of $W_i=\emptyset$ reduces to non-conditional tripartite information, $- \mi_3(X_i:Y_i:Z_i)$.
Notice the coefficient of  $- \mi_3(X_i:Y_i:Z_i|W_i)$ is always +1; other positive integer coefficients can be attained simply by repeating terms, whereas negative coefficients would \emph{not} constitute a TF. As an empirical observation, all known superbalanced primitive sHIQs are TF-compatible, leading to the following conjecture \cite{HubenyLiuWIP}:
\begin{conjecture}\label{conj:sHEIisTF}
All superbalanced primitive sHIQs are TF-compatible.
\end{conjecture}
An immediate corollary of \cref{conj:sHEIisTF} is that any superbalanced sHIQ is a conical combination of $-\mi_3$s. Notice that the converse of conjecture \ref{conj:sHEIisTF} is clearly not true, a counterexample being $- \mi_3(X:Y:Z|W)$ itself, as exemplified by two holographic graph models \cite{Hernandez-Cuenca:2023iqh}.

\subsubsection{Effect of null reduction}

The TF of an information quantity $\QQ$ is particularly well-suited to the present context, since we can analyze what happens to $\QQ$ under null reductions by considering what happens to each term individually. Suppose for concreteness that we are null-reducing on the party $A$, and let us consider a single term in the TF. Since $\mi_3(X:Y:Z|W)$ is symmetric under permuting $X,Y,Z$, there are three distinct possibilites for which subsystem can contain the central region $\XX$:
\begin{enumerate}
    \item $\XX$ is contained in one of $X$, $Y$, or $Z$: Suppose (without loss of generality) that $\XX\subseteq X$.  Then 
\begin{align}\label{eq:nlr-firstcase}
\begin{split}
    -\ANR{\mi_3(X:Y:Z|W)}{\XX} &= -\ANR{\mi_3(WX:Y:Z)}{\XX} + \ANR{\mi_3(W:Y:Z)}{\XX}
    \\
    &= - \ANR{\mi_3(WX:Y:Z)}{\XX}\\
    &= \mi(Y:Z|WX)
\end{split}
\end{align}
where in the first line we rewrote the conditional tripartite information as a difference of two tripartite information quantities as in  \eqref{eq:cond-tripartiteI}, in the second line we dropped the second term in which $X$, and therefore $\XX$, is absent, and in the third line we used \eqref{MMINRA} to actually perform  the null reduction.  Similarly if $\XX\subseteq Y$, then $-\ANR{\mi_3(X:Y:Z|W)}{\XX}=\mi(X:Z|WY)$, and if $\XX\subseteq Z$, then $-\ANR{\mi_3(X:Y:Z|W)}{\XX}=\mi(X:Y|WZ)$.  In all cases, the subsystem containing $\XX$ joins the conditioned-on argument in the CMI. 
    \item $\XX$ is contained in $W$: Then \begin{align}\label{eq:nlr-secondcase}
\begin{split}
    -\ANR{\mi_3(X:Y:Z|W)}{\XX} &= -\ANR{(\mi(Y:Z|WX) - \mi(Y:Z|W))}{\XX}\\
    &= -\mi_3(X:Y:Z|W)\,,
\end{split}
\end{align}
so we simply get back the original conditional tripartite information. Here the first line comes from conditioning the following identity on $W$,
\begin{equation}
    -\mi_3(X:Y:Z) = \mi(Y:Z|X) - \mi(Y:Z)\,,
\end{equation}
and the second line follows from noticing that all terms contain $W$ so the null reduction operator acts as the identity. 

    \item $\XX$ is not contained in any of $X$,$Y$,$Z$,$W$: The null reduction kills the entire term,
\be\label{eq:nlr-thirdcase}
-\ANR{\mi_3(X:Y:Z|W)}{\XX} = 0\,.
\ee
\end{enumerate}

Since we are interested in the effect of the null reduction on the entire $\QQ$, it is natural to compare $\ANR{\QQ}{\XX}$, when evaluated on an entropy vector in the RT cone, both with $\QQ$ and with 0. The above results tell us what happens to each term individually. In the first case ($\XX \subseteq X$), we can re-express $-\ANR{\mi_3(X:Y:Z|W)}{\XX}$ more conveniently by recognizing that the CMI generated by null reduction \eqref{eq:nlr-firstcase} is in fact the original conditional tripartite information term plus another CMI, 
\begin{equation}
    \mi(Y:Z|WX)= -\mi_3(X:Y:Z|W) + \mi(Y:Z|W)\,.
\end{equation}
Hence, by SSA, 
this null reduction can only increase the value of the term, and is simultaneously guaranteed to be non-negative.
In the second case  ($\XX \subseteq W$), the null reduction \eqref{eq:nlr-secondcase} preserves the value, which itself is however not sign-definite. In the third case  ($\XX \nsubseteq XYZW$),  the null reduction \eqref{eq:nlr-thirdcase} is itself non-negative (since it's simply 0), but now it is no longer clear whether this operation increases or decreases the value of the original term.

Hence, for a generic information quantity written in TF, any given null reduction is \emph{not} guaranteed to be either positive or greater than $\QQ$ itself. Nonetheless, as we will now discuss, if more information is available about the information quantity, then powerful consequences can be derived from the TF.

\subsubsection{Restricted tripartite form}

Suppose now that the entire TF expression for $\QQ$ has the property that none of the $W_i$'s contain $\XX$. Then we are guaranteed that $\ANR{\QQ}{\XX}$ is just a sum of CMIs:
\begin{equation}\label{eq:NLRforRTF}
    \ANR{\QQ}{\XX} = \sum_j \mi(\tilde{X}_j:\tilde{Y}_j|\tilde{W}_j) \ge0\, .
\end{equation}
This automatically renders it both an sHIQ and majorizable. The former follows by SSA, while the latter follows from the fact that the majorization test is preserved by conical combinations and the fact that each term in \eqref{eq:NLRforRTF} individually majorizes, cf.\ \eqref{MMImajorize}. Another nice possibility is that $\XX \subset X_iY_iZ_iW_i$ for each $i$, i.e.\ all terms fall into the first two categories.  Then we can rewrite $\ANR{\QQ}{\XX}$ as the sum of the original $\QQ$, along with CMIs, which by SSA means that $\ANR{\QQ}{\XX}(\vec{\mathsf{S}}) \ge \QQ(\vec{\mathsf{S}})$ for any holographic (or even quantum) entropy vector, in other words that  $-\ANRC{\QQ}{\XX} = \ANR{\QQ}{\XX} - \QQ$ is an sHIQ (independently of whether or not $\QQ$ is).  If additionally $\QQ$ is in fact an sHIQ, the null reduction $ \ANR{\QQ}{\XX} $ will likewise be one. The sHIQ $\QQ^{[5]}$ studied in Example \ref{Q5eg} has this property.

The previous paragraph made assumptions about the collection of terms in the TF.  When considering a given primitive sHIQ, such as specified for $\N=6$ in \cite{Hernandez-Cuenca:2023iqh}, it is often  the case (scanning over all $\QQ$'s) that each specific party $X$ is absent from some term $i$ in the TF, the conditioned-on parties $\{W_i\}$ are distinct (the number of these ranges from 0 to 5), and so there are also some parties which are not conditioned on by any term.  However, this by itself does not manifest LCC-safety.  To do better, we would need to show that we can re-express $\QQ$ in different versions of TF, such that for each party, we can find a version where that party is not conditioned-on by any term.    

To motivate this possibility, note that while TF-compatibility certainly poses a restriction on an information quantity, for a given TF-compatible $\QQ$, the form \eqref{eq:tripartite-form-def} is far from unique (with the exception of a sum of non-conditioned $\mi_3$'s with elementary subsystem arguments).  This is evident from the fact that the number of distinct terms is exponentially (in $\N$) larger than the dimensionality of the entropy space.  For example, by rewriting \eqref{eq:cond-tripartiteI}, we have three natural TF expressions:
\begin{equation}\label{eq:rewriteTF}
    \mi_3(WX:Y:Z) 
        = \mi_3(W:Y:Z) +\mi_3(X:Y:Z|W) 
        =  \mi_3(X:Y:Z) +\mi_3(W:Y:Z|X)  \, ,
\end{equation}
and if any of the arguments are composite, we have many more. We can then use these types of manipulations to recast a given $\QQ$ from one TF to another, to shift which parties are conditioned on; for example in \eqref{eq:rewriteTF} the last equality shifted the conditioned-on party from $W$ to $X$.

We say an information quantity $\QQ$ can be written in a \emph{restricted tripartite form} or is \emph{RTF-compatible} if, for any choice of a specific party $V$,  it can be written in a tripartite form \eqref{eq:tripartite-form-def}, such that none of the $W_i$'s contain $V$. The utility of RTF-compatibility stems from the fact that any RTF-compatible sHIQ is LCC-safe, and all its null reductions are sHIQ. If RTF-compatibility held for all sHIQs, this would immediately prove both \cref{conj:sHIQmaj} and \cref{conj:QisTrue}.  Moreover, it would give a useful hint regarding a potentially more insightful repackaging of sHEIs.  However, it is \emph{not} the case that all sHIQs have this property, as we saw in the case of $\mathbf{Q}^{[7]}$ in Example \ref{Q7eg}.

\subsubsection{Conjecture \ref{conj:sHEIisTF} implies conjecture \ref{conj:QisTrue}}

Let us now consider the consequences of TF-com\-pat\-i\-bil\-i\-ty upheld for all sHIQs. If we assume \cref{conj:sHEIisTF}, and additionally assume contraction map completeness (namely that if $\QQ$ is an sHIQ then it has a contraction map), as argued in \cite{Bao:2025sjn}, we can actually prove \cref{conj:QisTrue}.  This subsubsection will be devoted to the proof.

\begin{thm}
Let $\QQ$ be a contraction-provable and TF-compatible sHIQ. Then all of its null reductions are sHIQs.
\end{thm}

\begin{proof}
Let $f$ be the contraction map for $\QQ$, characterized by \eqref{eq:contrOV} and \eqref{eq:contrcond}.  It can be expressed as a table with $L$ LHS columns and $R$ RHS columns, with LHS rows comprising all possible $\{0,1\}^L$ bit strings, and RHS rows with corresponding  $\{0,1\}^R$ bit strings such that any pair has no greater Hamming distance than corresponding LHS pair.\footnote{
    For convenience we assume all coefficients are 1, which can be achieved by writing out the terms with higher coefficients multiple times.} 
The occurrence vector pairs have equal Hamming distance on both sides.

Consider any party, say $\XX$, and the corresponding null reduction $\ANR{\QQ}{\XX}$.  We now propose a contraction map for $\ANR{\QQ}{\XX}$ as follows: 
\begin{enumerate} 
\item Drop LHS and RHS columns corresponding to $\ANRC{\QQ}{\XX}$, i.e.\ all terms which don't contain $\XX$.  This creates a degeneracy amongst the LHS bit strings, which we fix as follows:
\item Keep all occurrence vectors, and drop any other row whose LHS bit string was degenerate with any of these.  
\item From the remainder, delete all rows which had non-zero entries in the deleted LHS columns.  In other words, apart from the occurrence vectors, the other LHS bit strings are those which had 0's in the deleted columns.  
\end{enumerate} 
We call this new map $\fnr{\XX}$. 
It now has $N$ LHS and $N$ RHS columns, and $2^N$ rows corresponding to all LHS bit strings $\vec x \in \{0,1\}^N$.
To show that $\fnr{\XX}$ is indeed a contraction map, we need to show that for any pair of bit strings  $\vec x,\vec x' \in \{0,1\}^N$, the RHS Hamming distance is no larger than the LHS one, 
$\norm{\fnr{\XX}(\vec x) - \fnr{\XX}(\vec x')}_{\text{rhs}} \leq \norm{\vec x-\vec x'}_{\text{lhs}}$.

There are three types of possibilities:
\begin{enumerate} 
\item $\vec x$ and $\vec x'$ are both occurrence vectors
\item only one of $\vec x$ and $\vec x'$ is an occurrence vector
\item neither $\vec x$ nor $\vec x'$ are occurrence vectors.
\end{enumerate} 
In the first case, suppose we take the occurrence vectors corresponding to elementary subsystems $X$  and $Y$.  The Hamming distance is given by the number of terms which have $X$ but not $Y$, or those which have $Y$ but not $X$ (since otherwise if they have both $X$ and  $Y$, we get $1-1=0$ contribution, and if they have neither, we likewise get $0-0=0$ contribution). Let us denote this by $\#_{X\wedge Y}$.  As a warm-up, let us first see how to recover the fact that this number is the same between LHS and RHS for the original $\QQ$, using the tripartite form expression.  We can do this term by term.  Since we're assuming $\QQ$ is TF-compatible, it suffices to consider each $-\mi_3(X_i:Y_i:Z_i|W_i)$ term individually, and then combine them.  For concreteness, suppose one of these terms is 
\begin{multline}
-\mi_3(A:B:C|D) = \\
\ent{D} + \ent{ABD} + \ent{ACD} + \ent{BCD} - \ent{AD} - \ent{BD} - \ent{CD} - \ent{ABCD}  \,.
\end{multline}
The distinct types of possibilities for $X$ and $Y$ are (up to $\{A,B,C\}$ permutations or $\{E,F,\ldots\}$ permutations) $A \wedge B$, $A \wedge D$, $A \wedge E$, and $D \wedge E$.  It is trivial to check explicitly that for each of these possibilities 
\begin{equation}
 \Delta \#_{X\wedge Y} \coloneq
    \#_{X\wedge Y}(LHS) - \#_{X\wedge Y}(RHS) = 0 \, .
\end{equation}
The same result will then hold for each $-\mi_3(X_i:Y_i:Z_i|W_i)$ and therefore their sum $\QQ$.
Now consider the null reduction $\ANR{\QQ}{\XX}$.  For each term $-\mi_3(X_i:Y_i:Z_i|W_i)$, the null reduction either gives conditional mutual information (if $\XX \subset X_iY_iZ_i$), remains the same (if $\XX \subseteq W_i$), or else gives zero.  In the latter two cases we still have $\Delta \#_{X\wedge Y}= 0$.  In the first case, say (for concreteness) we consider $\mi(B:C|AD)$.  Then  $\Delta \#_{B\wedge C} = 2$, while all the other possibilities vanish. The upshot is that for any $X$ and $Y$, and any null-reduced TF-compatible sHIQ, $\Delta \#_{X\wedge Y} \ge 0$.  This means that the map $\fnr{\XX}$ is a contraction on the occurrence vectors.

To analyze the other two types of cases, we first make an observation: In the original contraction map $f$, all rows with LHS bit strings having 0's for all terms in $\ANRC{\QQ}{\XX}$ necessarily have the RHS terms corresponding to $\ANRC{\QQ}{\XX}$ likewise all 0's. To see this, consider first two specific occurrence vectors: $\vec{x}_{\XX}$, for which $f$ has 1's in all $\ANR{\QQ}{\XX}$ columns and 0's in all $\ANRC{\QQ}{\XX}$ columns, and $\vec{x}_{\N+1}$, which has 0's in all columns.  Now consider any other occurrence vector $\vec{x}$ with 0's in all $\ANRC{\QQ}{\XX}$ columns.  Suppose it has  $n_0$ 0's and $n_1$ 1's in the $\ANR{\QQ}{\XX}$ columns.  Since $f$ is a contraction map, we know that the corresponding RHS bit string can have no more than $n_1$ 1's (since $\norm{f(\vec x) - f(\vec{x}_{\N+1})}_{\text{rhs}} \leq \norm{\vec x-\vec{x}_{\N+1}}_{\text{lhs}} = n_1$) and no more than $n_0$ 0's (since $\norm{f(\vec x) - f(\vec{x}_{\XX})}_{\text{rhs}} \leq \norm{\vec x-\vec{x}_{\XX}}_{\text{lhs}} = n_0$) already in the $\ANR{\QQ}{\XX}$ columns.  Since $n_0+n_1=N$, which is the number of both LHS and RHS  $\ANR{\QQ}{\XX}$ columns, and contraction is already saturated on these, in order for $f$ to remain a contraction map, we need  $\norm{f(\vec x) - f(\vec{x}_{\XX})}_{\text{rhs}} = \norm{f(\vec x) - f(\vec{x}_{\N+1})}_{\text{rhs}} = 0$ on all $\ANRC{\QQ}{\XX}$ columns, so all entries in $f(\vec{x})$ must be 0 there.

It is now easy to see that $\fnr{\XX}$ is a contraction map for cases 2 and 3 as well.  In case 3, the Hamming distance remains the same on both LHS and RHS, i.e.\  $\norm{\fnr{\XX}(\vec x) - \fnr{\XX}(\vec x')}_{\text{rhs}} = \norm{f(\vec x) - f(\vec x')}_{\text{rhs}}  \leq \norm{\vec x-\vec x'}_{\text{lhs}}$ (since all deleted columns had 0's for both $\vec{x}$ and $\vec{x}'$ rows).  In case 2, the Hamming distance in the $\ANRC{\QQ}{\XX}$ columns is given by the number of 1's in the occurrence vector, which by balance of $\ANRC{\QQ}{\XX}$ is the same on both LHS and RHS.  Since the original $f$ was contraction map and the Hamming distance decreases by the same amount on both sides, $\fnr{\XX}$ remain a contraction map.

We have now proved that in all cases, $\fnr{\XX}$ is a contraction map.  This guarantees that $\ANR{\QQ}{\XX}$ is an sHIQ.  Since the choice of $\XX$ was arbitrary, repeating the argument for all regions, we see that all null reductions of $\QQ$ are sHIQs.
\end{proof}

\acknowledgments

We would like to thank Bartek Czech, Brianna Grado-White, Sarah Harrison, Sergio Hern\'an\-dez-Cuenca, Sami Kaya, Massimiliano Rota, and Fabian Ruehle for useful conversations. G.G. and M.H. were supported by the Department of Energy through award DE-SC0009986. V.H. was supported in part by the Department of Energy award DE-SC0009999 and by funds from the University of California. This research was supported in part by grant NSF PHY-2309135 to the Kavli Institute for Theoretical Physics (KITP), where part of this work was completed. This work was also performed in part at the Aspen Center for Physics, which is supported by National Science Foundation grant PHY-2210452. We are also grateful to the Centro de Ciencias de Benasque Pedro Pascual, where part of this work was completed.

\appendix
\section{Review of majorization}
\label{sec:review-majorization}

Consider two $N$ dimensional real vectors $\vec{x}$ and $\vec{y}$ ordered component-wise, i.e.\ such that $x_1 \geq x_2 \geq \dots \geq x_N$ and $y_1 \geq y_2 \geq \dots \geq y_N$. If the sum of the first $k$ components of $\vec{x}$ is less than or equal to the sum of the first $k$ components of $\vec{y}$ for all $k = 1, \dots, N-1$ and if the sums are equal for $k = N$, then we say that the vector $\vec{x}$ is majorized by $\vec{y}$ (in notation $\vec{x} \prec \vec{y}$). More compactly, $\vec x\prec\vec y$ if
\begin{equation}\label{majdef}
 \sum_{n = 1}^{k} x_n\le \sum_{n = 1}^{k} y_n 
  \quad( k = 1,2,\dots, N-1)\,,\qquad
  \sum_{n = 1}^{N} x_n=\sum_{n = 1}^{N} y_n \,.
\end{equation}
Majorization comes with an ordering of the components; if the two vectors are initially not ordered then one vector is majorized by the other if it is majorized once the components of both vectors are ordered. Intuitively, majorization is a statement about how much more ``spread out" the components of a vector are with respect to those of another vector. For example, the vector $\vec{x} = (3,4)$ is majorized by the vector $\vec{y} =(2,5)$, since once we order their components we have $4\le 5$ and  $4+3 = 5+2$.  

There are many equivalent definitions of majorization. One is through doubly stochastic matrices: $\vec{x}$ is majorized by $\vec{y}$ if and only if $\vec{x}$ is in the convex hull of all the vectors obtained from permuting the components of $\vec{y}$. This is expressed formally as the existence of a doubly stochastic matrix $\mathsf{D}$ (i.e.\ a square matrix with non-negative entries whose rows and columns all sum to 1) such that
\begin{equation}
    \vec{x} = \mathsf{D}\vec{y}.
\end{equation}

Another equivalent definition is through inequalities of concave (or convex) functions: $\vec{x}$ is majorized by $\vec{y}$ if and only if, for any concave function $\wdef: \mathbb{R} \to \mathbb{R}$, one has
\begin{equation}
    \sum_{n = 1}^{N} \wdef(x_n) \geq \sum_{n = 1}^{N} \wdef(y_n)\,.
\end{equation}
The inequality is reversed for convex functions, since if $\wdef$ is concave then $-\wdef$ is convex. This definition of majorization is extremely powerful, since it means that, to guarantee that $\vec{x} \prec \vec{y}$, it suffices to check that (i) $\sum_n x_n = \sum_n y_n$, and (ii) for all $t\in\R$,
\begin{equation}\label{majmax}
    \sum_{n=1}^{N} \max(0, x_n - t) \leq \sum_{n=1}^{N} \max(0, y_n- t)\,.
\end{equation}
This is the definition we use to implement the majorization test for our holographic entropy inequalities.

We now discuss weighted majorization. Let $\vec\alpha$, $\vec x$ be vectors of length $M$ and $\vec\beta$, $\vec y$ be vectors of length $N$, such that $\vec\alpha$, $\vec\beta$ have positive components and
\be
\sum_{m=1}^M\alpha_m=\sum_{n=1}^N\beta_n\,,\qquad
\sum_{m=1}^M\alpha_mx_m=\sum_{n=1}^N\beta_ny_n\,.
\ee
We say $(\vec x,\vec\alpha)\prec (\vec y,\vec\beta)$ if, for all $t\in\R$,
\be
\sum_{m=1}^M\alpha_m\max(0,x_m-t)\le
\sum_{n=1}^N\beta_n\max(0,y_n-t)\,.
\ee
This is equivalent to the statement that, for any concave function $h:\R\to\R$,
\be
\sum_{m=1}^M\alpha_mh(x_m)\ge\sum_{n=1}^N\beta_nh(y_n)\,.
\ee
It is also equivalent to the existence of an $M\times N$ matrix $\mathsf{D}$ with non-negative entries such that
\be
\vec 1_M=\mathsf{D}\vec 1_N\,,\qquad
\vec\alpha\mathsf{D}=\vec\beta\,,\qquad
\vec x=\mathsf{D}\vec y\,,
\ee
where $\vec 1_M$ is the length-$M$ vector with all components equal to 1 \cite{sherman1951theorem,borcea2007equilibrium}.

\section{Example code}\label{sec:code}

Here we expand briefly on how the conjectures were tested, specifically how to implement both the analytic and numerical majorization tests in \emph{Mathematica}.  

The first step is to generate the necessary data, i.e.\ all the null reductions of a given set of inequalities that we want to test. We will not show this here. However, our starting point was to upload in Mathematica all the available sHIQs in the form of entropy vectors, which are publicly available here \cite{hecdata}. We will present the code by using the following $\mathsf{N} = 5$ primitive sHIQ as a concrete example:
\begin{multline}
    \mathbf{Q}^{[5]} = \ent{ABC} + \ent{ABD}  + \ent{ACE}  + \ent{BCD} + \ent{BCE}\\
    - \ent{A} - \ent{BC} - \ent{BD} - \ent{CE} - \ent{ABCD} - \ent{ABCE}\,.
\end{multline}
This sHIQ will have five null reductions, and for each, there will be a pair of LHS and RHS vectors to be tested for majorization. So, we store the necessary information in the following list:
\begin{mmaCell}{Code}
  Q5 = {{{a+b+c, a+b+d, a+c+e},{a,a+b+c+d, a+b+c+e}},
{{a+b+c, a+b+d, b+c+d, b+c+e},{b+c, b+d, a+b+c+d, a+b+c+e}},
{{a+b+c, a+c+e, b+c+d, b+c+e},{b+c, c+e, a+b+c+d, a+b+c+e}},
{{a+b+d, b+c+d},{b+d, a+b+c+d}}, {{a+c+e, b+c+e},{c+e, a+b+c+e}}};
\end{mmaCell}
The list above is storing all five pairs of vectors in the form \{LHS vector, RHS vector\}.

Performing the majorization test algebraically is then easy using \eqref{majmax}, by writing the following function \texttt{majorizationTest} which takes the list Q5 above as input:
\begin{mmaCell}{Code}
majorizationTest[\mmaPat{majVecs_}] := Module[{\mmaPat{testTable}},
  \mmaPat{testTable} = Table[
    FullSimplify[
     Sum[Max[\mmaPat{majVecs}[[X, 1, \mmaFnc{i}]]-t, 0], {\mmaFnc{i}, 1, Length[\mmaPat{majVecs}[[X, 1]]]}] <=
     Sum[Max[\mmaPat{majVecs}[[X, 2, \mmaFnc{i}]]-t, 0], {\mmaFnc{i}, 1, Length[\mmaPat{majVecs}[[X, 2]]]}], 
     Assumptions -> {a>0, b>0, c>0, d>0, e>0, t \[Element] Reals}], {X, 1, 5}];
  AllTrue[\mmaPat{testTable}, TrueQ]
  ]
\end{mmaCell}
The function returns \texttt{True} if the inequality is LCC-safe, i.e.\ if every null reduction passes the majorization test. (The above code is suitable for testing an inequality involving five regions. To extend it to more regions, simply introduce additional variables in \texttt{Assumptions} and increase the range of the variable X from 5 to the number of parties involved.)

If one instead wants to perform a numerical test, the following two functions suffice:
\begin{mmaCell}{Code}
Majorized[\mmaPat{v_List}] := 
 \mmaDef{VectorLessEqual}[\mmaDef{Accumulate} /@ Reverse /@ \mmaDef{NumericalSort} /@ \mmaPat{v}] &&
 (Equal @@ Total /@ \mmaPat{v}) /; \mmaDef{AllTrue}[\mmaPat{v}, NumericQ, 2] 
\end{mmaCell}
This takes as input a list of two numerical vectors and checks whether the first vector is majorized by the second vector.
\begin{mmaCell}{Code}
majorizationTestRandom[\mmaPat{majVecs_}, \mmaPat{n_}] := Module[{\mmaPat{testTable}},
  \mmaPat{testTable} = Table[
  Table[
     Majorized[\mmaPat{majVecs}[[\mmaFnc{X}]]] /. {a -> \mmaDef{RandomReal}[{0, 10}], 
       b -> \mmaDef{RandomReal}[{0, 10}], c -> \mmaDef{RandomReal}[{0, 10}],
       d -> \mmaDef{RandomReal}[{0, 10}], e -> \mmaDef{RandomReal}[{0, 10}], \mmaPat{n}], {\mmaFnc{X}, 1, 5}];
  \mmaDef{AllTrue}[Flatten[\mmaPat{testTable}], TrueQ]
  ]
\end{mmaCell}
This takes two inputs, the list of null reductions and a positive integer $n$ which determines the total number of independent random draws of the variables $a,b,c,\dots$ to be performed for each null reduction. The function returns \texttt{True} if no counterexamples are found. This method cannot prove LCC-safety, but if the test is passed for large enough $n$ (such as $\sim100$) the inequality is established as LCC-safe beyond any reasonable doubt. This function has the advantage of being extremely fast: one can test all known (non-MI) 1876 $\mathsf{N} = 6$ primitive sHIQs orbits with $n = 100$ on a laptop in less than a minute:
\begin{mmaCell}{Code}
\mmaDef{testAllHIQs} = Table[\mmaDef{majorizationTestRandom}[\mmaDef{HIQ}[[\mmaFnc{k}]], 100], {\mmaFnc{k}, 2, 1877}];
\mmaDef{AbsoluteTiming}[\mmaDef{AllTrue}[\mmaDef{testAllHIQs}, TrueQ]]
\end{mmaCell}
\begin{mmaCell}{Output}
\{28.2506, True\}
\end{mmaCell}

Finally, to check whether an information quantity is an sHIQ (needed to test conjectures \ref{conj:QisTrue} and \ref{conj:allNR}) one can use the following simple function \texttt{InConeQ}:
\begin{mmaCell}{Code}
InConeQ[\mmaPat{cone_List}, \mmaPat{vec_List}] := Module[{\mmaPat{cc}, \mmaPat{sol}},
  \mmaPat{cc} = Array[c, Length[\mmaPat{cone}]];  
  \mmaPat{sol} = FindInstance[Join[Thread[\mmaFnc{cc}.\mmaPat{cone} == \mmaPat{vec} && Thread[\mmaFnc{cc} >= 0]]],\mmaFnc{cc}];
  If[\mmaPat{sol} == {}, False, True]
  ]
\end{mmaCell}
The function takes as input (i) a list, \texttt{cone}, which contains all extreme rays of the dual cone (for $\mathsf{N} = 5$, as in sec. \ref{sec:results}, these are the 372 superbalanced sHIQs), and (ii) the vector \texttt{vec}, corresponding to the information quantity that we want to test. The function then checks whether \texttt{vec} can be written as a conical combination of the basis vectors in \texttt{cone}. If a solution exists, the information quantity is confirmed to be a sHIQ, and the function returns \texttt{True}.

\bibliographystyle{jhep}
\bibliography{references}
\end{document}